\documentclass[reqno]{my_statsoc}

\usepackage[a4paper]{geometry} % required for statsoc+pdflatex

\usepackage{amsmath,amssymb,amsthm}
\usepackage[colorlinks,citecolor=blue,urlcolor=blue,breaklinks]{hyperref}
\usepackage{natbib}
\usepackage{pdfsync}
\usepackage{float}
\usepackage{graphicx}

\usepackage{footmisc}

\RequirePackage{hypernat}

\usepackage{amsmath,examplep,wasysym}

\usepackage{amsfonts,mathrsfs}
\usepackage{dsfont}

\DeclareMathOperator*{\argmin}{arg\,min}

\newcommand{\RR}{\mathbb{R}}

\newcommand{\NN}{\mathbb{N}}

\newcommand{\PP}{\mathbb{P}}

\newcommand{\eps}{\epsilon}

\newcommand{\bzeta}{\boldsymbol{\zeta}}
\newcommand{\btheta}{\boldsymbol{\theta}}
\newcommand{\bvartheta}{\boldsymbol{\vartheta}}
\newcommand{\bthetaML}{\boldsymbol{\theta}^{{\rm ML}}}
\newcommand{\bthetaB}{\boldsymbol{\theta}^{{\rm B}}}

\newcommand{\bxi}{\boldsymbol{\xi}}
\newcommand{\bfeta}{{\boldsymbol{\eta}}}
\newcommand{\bfX}{\mathbf X}
\newcommand{\bfA}{\mathbf A}
\newcommand{\bfB}{\mathbf B}
\newcommand{\bfE}{\mathbf E}
\newcommand{\bfH}{\mathbf H}

\newcommand{\bfP}{\mathbf P}
\newcommand{\bfM}{\mathbf M}
\newcommand{\bfQ}{\mathbf Q}
\newcommand{\bfI}{\mathbf I}

\newcommand{\bfSigma}{\boldsymbol\Sigma}

\newcommand{\ba}{\mathbf a}
\newcommand{\boldb}{\boldsymbol b}

\newcommand{\bu}{\boldsymbol u}
\newcommand{\bv}{\boldsymbol v}

\newcommand{\bW}{\boldsymbol W}

\newcommand{\bD}{\boldsymbol D}
\newcommand{\bL}{\boldsymbol L}
\newcommand{\bP}{\bfP}
\newcommand{\bPLt}{\bfP^t_{\!\!\bL}}

\newcommand{\bPDh}{\bfP^h_{\!\!\bD}}
\newcommand{\bX}{\boldsymbol X}
\newcommand{\bx}{\boldsymbol x}
\newcommand{\bY}{\boldsymbol Y}
\newcommand{\bZ}{\boldsymbol Z}
\newcommand{\by}{\boldsymbol y}

\newcommand{\ind}{\mathds 1}
\newcommand{\mean}{\mathop{\rm mean}}

\newtheorem{theorem}{Theorem}
\newtheorem{proposition}{Proposition}
\newtheorem{lemma}{Lemma}
\newtheorem{corollary}{Corollary}

\newtheorem{remark}{Remark}

\def\toprule{\hline}
\def\midrule{\hline}
\def\bottomrule{\hline}

\title[Sampling from smooth and log-concave densities]{Theoretical guarantees for approximate sampling from\\ smooth and log-concave densities}
\author[A. S. Dalalyan]{Arnak S. Dalalyan}
\address{ENSAE ParisTech - CREST,\\
3, Avenue Pierre Larousse,\\
92240 Malakoff, France.}
\email{arnak.dalalyan@ensae.fr}

\begin{document}

\begin{abstract}
Sampling from various kinds of distributions is an issue of pa\-ra\-mount  importance in statistics
since it is often the key ingredient for constructing estimators, test procedures or
confidence intervals. In many situations, the exact sampling from a given distribution
is impossible or computationally expensive and, therefore, one needs to resort to approximate
sampling strategies. However, there is no well-developed theory providing meaningful
nonasymptotic guarantees for the approximate sampling procedures, especially in the
high-dimensional problems. This paper makes some progress in this direction
by considering the problem of sampling from  a distribution having a smooth and log-concave
density defined on $\RR^p$, for some integer $p>0$. We establish nonasymptotic bounds for the
error of approximating the target distribution by the one obtained by the Langevin Monte Carlo
method and its variants. We illustrate the effectiveness of the established guarantees with 
various experiments. Underlying our analysis are insights from the theory of continuous-time 
diffusion processes, which may be of interest beyond the framework of log-concave densities 
considered in the present work.
\end{abstract}

\vspace{-20pt}


%\maketitle
%\dominitoc

%\setcounter{tocdepth}{1}
%\tableofcontents

\section{Introduction}

Let $p$ be a positive integer and $f:\RR^p\to\RR$ be a measurable function such that the integral $\int_{\RR^p} \exp\{-f(\btheta)\}\,d\btheta$
is finite. 
If we think of $f$ as the negative log-likelihood or the negative log-posterior of a statistical model, then the maximum likelihood
and the Bayesian estimators, which are perhaps the most popular in statistics, are respectively defined as
$$
\bthetaML \in \argmin_{\btheta\in\RR^p} f(\btheta);\qquad
\bthetaB = \frac{1}{\int_{\RR^p} e^{-f(\bu)}\,d\bu}\int_{\RR^p} \btheta e^{-f(\btheta)}\,d\btheta.
$$
These estimators are rarely available in closed-form. Therefore, optimisation techniques are used for computing the maximum-likelihood
estimator while the computation of the Bayes estimator often requires sampling from a density proportional to $e^{-f(\btheta)}$.
In most situations, the exact computation of these two estimators is impossible and one has to resort to approximations provided by
iterative algorithms. There is a vast variety of such algorithms for solving both tasks, see for example
\citep{BoydBook} for optimisation and \citep{Atchade2011} for approximate sampling. However, a striking fact is that the convergence
properties of optimisation algorithms are much better understood than those of the approximate sampling algorithms. The goal of
the present work is to partially fill this gap by establishing easy-to-apply theoretical guarantees for some
approximate sampling algorithms.

To be more precise, let us consider the case of a strongly convex function $f$ having a Lipschitz continuous gradient. That is, there
exist two positive constants $m$ and $M$ such that
\begin{equation} \label{convex1}
\begin{cases}
f(\btheta)-f(\bar\btheta)-\nabla f(\bar\btheta)^\top (\btheta-\bar\btheta)
\ge \frac{m}2\|\btheta-\bar\btheta\|_2^2,\\
\|\nabla f(\btheta)-\nabla f(\bar\btheta)\|_2 \le M \|\btheta-\bar\btheta\|_2,
\end{cases}
\qquad \forall \btheta,\bar\btheta\in\RR^p,
\end{equation}
where $\nabla f$ stands for the gradient of $f$ and $\|\cdot\|_2$ is the Euclidean norm. There is a simple result characterising the
convergence of the  well-known gradient descent algorithm under the assumption (\ref{convex1}).

\begin{theorem}[Eq. (9.18) in \citep{BoydBook}]\label{th:1}
If $f:\RR^p\to\RR$ is continuously differentiable and fulfils (\ref{convex1}), then the gradient descent algorithm defined recursively by
\begin{align}\label{algoGD}
\btheta^{(k+1)} = \btheta^{(k)} - (2M)^{-1} \nabla f(\btheta^{(k)});\qquad k=0,1,2,\ldots
\end{align}
satisfies
\begin{equation}\label{eq:th1}
\|\btheta^{(k)}-\bthetaML\|_2^2 \le \frac{2\big(f(\btheta^{(0)})-f(\bthetaML)\big)}m \Big(1-\frac{m}{2M}\Big)^{k},\qquad \forall k\in\NN.
\end{equation}
\end{theorem}
This theorem implies that the convergence of the gradient descent is exponential in $k$. More precisely, it
results from Eq.\ (\ref{eq:th1}) that in order to achieve an approximation error upper bounded by $\epsilon>0$ 
in the Euclidean norm it suffices to perform
\begin{equation}\label{keps}
k_\epsilon = \frac{\log \big\{2m^{-1}\big(f(\btheta^{(0)})-f(\bthetaML)\big)\big\}+2\log(1/\epsilon)}{\log \frac{2M}{2M-m}}
\end{equation}
evaluations of the gradient of $f$. An important feature of this result is the logarithmic dependence of $k_\epsilon$ on
$\epsilon$ but also its independence of the dimension $p$. Note also that even though the right-hand side of (\ref{keps}) is a somewhat
conservative bound on the number of iterations, all the quantities involved in that expression are easily computable and lead to a
simple stopping rule for the iterative algorithm.

The situation for approximate computation of $\bthetaB$ or for approximate sampling from the density proportional to $e^{-f(\btheta)}$
is much more contrasted. While there exist almost as many algorithms for performing these tasks as for the optimisation, the convergence
properties of most of them are studied only empirically and, therefore, provide little theoretically grounded guidance for the
choice of different tuning parameters or of the stopping rule. Furthermore, it is not clear how the rate of convergence of these algorithms
scales with the growing dimension. While it is intuitively understandable that the problem of sampling from a distribution is more difficult
than that of maximising its density, this does not necessarily justifies the huge gap that exists between the precision of theoretical 
guarantees available for the solutions of these two problems. This gap is even more surprising in light of the numerous similarities 
between the optimisation and approximate sampling algorithms.

Let us describe a particular example of approximate sampling algorithm, the Langevin Monte Carlo (LMC),
that will be studied throughout this work.
Its definition is similar to the gradient descent algorithm for optimisation but involves an additional
step of random perturbation. Starting from an initial point $\bvartheta^{(0)}\in\RR^p$ that may be
deterministic or random, the subsequent steps of the algorithm are defined by
the update rule
\begin{align}\label{algoLMC}
\bvartheta^{(k+1,h)} = \bvartheta^{(k,h)} - h \nabla f(\bvartheta^{(k,h)})+ \sqrt{2h}\;\bxi^{(k+1)};\qquad k=0,1,2,\ldots
\end{align}
where $h>0$ is a tuning parameter, often referred to as the step-size, and $\bxi^{(1)},\ldots,\bxi^{(k)},\ldots$
is a sequence of independent centered Gaussian vectors with
covariance matrix equal to identity and independent of $\bvartheta^{(0)}$. It is well known that under some assumptions on $f$, when 
$h$ is small and $k$ is large (so that the product $kh$ is large), the distribution of $\bvartheta^{(k,h)}$ is close in total 
variation to the distribution with density proportional to $e^{-f(\btheta)}$, hereafter referred to as the target distribution. 
The goal of the present work is to establish a  nonasymptotic upper bound, involving only explicit and computable quantities, 
on the total variation distance between the target distribution and its approximation by the distribution of $\bvartheta^{(k,h)}$. 
We will also analyse a variant of the LMC, termed LMCO, which makes use of the Hessian of $f$.

In order to give the reader a foretaste of the main contributions of the present work, we summarised in Table~\ref{tab:0}
some guarantees established and described in detail in the next sections. To keep things simple, we translated all the
nonasymptotic results into asymptotic ones for large dimension $p$ and small precision level $\epsilon$ (the $O^*$ notation
ignores the dependence on constant and logarithmic factors). The complexity of one iteration of the LMC indicated in the
table corresponds to the computation of the gradient and generation of a Gaussian $p$-vector, whereas the complexity
of one iteration of the LMCO is the cost of performing a singular values decomposition on the Hessian matrix of $f$, which
is of size $p\times p$.
%\newpage
\begin{table}
\caption{\label{tab:0} Summary of the main findings of this work. The first two columns provide the order of magnitude
of the number of iterates to perform in order to make the error of approximation smaller than $\epsilon$. The third column
contains the worst-case complexity of one iteration. Note that in many practical situations the real complexity might be much
smaller than the worst-case one.}
%\caption{\label{tab:0} Summary of the main findings}
\centering
\fbox{%
\begin{tabular}{l|c|c|c}
\toprule
					& number of iterates & number of iterates & complexity of \\
					& Gaussian start     & warm start         & one iteration\\					
					\midrule
LMC 			& $O^*(p^3\epsilon^{-2})$ & $O^*(p\epsilon^{-2})$ & $O(p)$\\
          & Theorem~\ref{th:2}    & Section~\ref{sec:4.1}  &       \\
					\midrule
LMCO      & $O^*(p^{5/2}\epsilon^{-1})$ & $O^*(p\epsilon^{-1})$ & $O(p^3)$\\
          & Theorem~\ref{th:3}        & Section~\ref{Ozaki}  &       \\
\bottomrule
\end{tabular}
}
\end{table}

\subsection{Notation}
For any $p\in\NN$ we write $\mathscr B(\RR^p)$ for the $\sigma$-algebra of Borel sets of $\RR^p$. The Euclidean norm
of $\RR^p$ is denoted by $\|\cdot\|_2$ while $\|\nu\|_{\rm TV}$ stands for the total variation norm of a signed measure
$\nu$: $\|\nu\|_{\rm TV} = \sup_{A\in\mathscr B(\RR^p)} |\nu(A)|$.
For two probability measures $\nu$ and $\bar\nu$ defined on a space $\mathcal X$ and such that $\nu$ is absolutely continuous with respect to $\bar\nu$,
the Kullback-Leibler and $\chi^2$ divergences between $\nu$ and $\bar\nu$ are respectively defined by
$$
\text{KL}(\nu\|\bar\nu) = \int_{\mathcal X} \log\Big(\frac{d\nu}{d\bar\nu}(\bx)\Big)\,\nu(d\bx)
\quad\text{and} \quad
\chi^2(\nu\|\bar\nu) = \int_{\mathcal X} \Big(\frac{d\nu}{d\bar\nu}(\bx)-1\Big)^2\,\bar\nu(d\bx).
$$
All the probability densities on $\RR^p$  are with respect to the Lebesgue measure, unless otherwise specified.
We denote by $\pi$ the probability density function proportional to $e^{-f(\btheta)}$, by $\bfP_\pi$ the corresponding
probability distribution and by $\bfE_\pi$ the expectation with respect to $\bfP_\pi$. For a probability density $\nu$
and a Markov kernel $\bfQ$, we denote by $\nu\bfQ$ the probability distribution $\big\{(\nu\bfQ)(A) = \int_{\RR^p}
\nu(\bx)\;\bfQ(\bx,A)\;d\bx: A\in\mathscr B(\RR^p)\big\}$. We say that the density $\pi(\btheta)\propto e^{-f(\btheta)}$
is log-concave (resp.\ strongly log-concave) if the function $f$ satisfies the first inequality of (\ref{convex1})
with $m=0$ (resp.\ $m>0$). We refer the interested reader to \citep{Saumard} for a comprehensive survey on log-concave
densities.

\section{Background on the Langevin Monte Carlo algorithm}

The rationale behind the LMC algorithm (\ref{algoLMC}) is simple: the Markov chain $\{\bvartheta^{(k,h)}\}_{k\in\NN}$ is the Euler discretisation of a continuous-time
diffusion process $\{\bL_t :t\in\RR_+\}$, known as Langevin diffusion, that has $\pi$ as invariant density. The Langevin diffusion is defined by the stochastic
differential equation 
\begin{equation}\label{L-SDE}
d\bL_t = -\nabla f(\bL_t)\,dt + \sqrt{2} \; d\bW_t,\qquad t\ge 0,
\end{equation}
where $\{\bW_t:t\ge 0\}$ is a $p$-dimensional Brownian motion. When $f$ satisfies condition (\ref{convex1}), equation (\ref{L-SDE}) has a unique
strong solution which is a Markov process. In what follows, the transition kernel of this process is denoted by $\bPLt (\bx,\:\cdot\:)$,
that is $\bPLt (\bx,A) = \bfP(\bL_t\in A|\bL_0 = \bx)$ for all Borel sets $A\subset\RR^p$ and any initial condition $\bx\in\RR^p$. Furthermore,
assumption (\ref{convex1}) yields the spectral gap property of the semigroup $\{\bPLt:\,t\in\RR_+\}$, which in turn implies that the
process $\bL_t$ is geometrically ergodic in the following sense.
\begin{lemma}\label{lem:1}
Under assumption (\ref{convex1}), for any probability density $\nu$,
\begin{equation}\label{ineq:lem1}
\|\nu\bPLt-\pi\|_{\rm TV}\le \frac12\chi^2(\nu\|\pi)^{1/2} e^{-{t m}/{2}} ,\qquad \forall t\ge 0.
\end{equation}
\end{lemma}
The proof of this lemma, postponed to Section~\ref{sec:proofs2}, is based on the bounds on the spectral gap established 
in~\cite[Remark 4.14]{ChenWang}, see also~\cite[Corollary 4.8.2]{BGL14}.  In simple words, inequality (\ref{ineq:lem1}) shows that
for large values of $t$, the distribution of $\bL_t$ approaches exponentially fast to the target distribution, and the idea behind
the LMC is to approximate $\bL_t$ by $\bvartheta^{(k,h)}$ for $t=kh$. Note that inequalities of type (\ref{ineq:lem1}) can be obtained
under conditions (such as the curvature-dimension condition, see \citet[Definition 1.16.1 and Theorem 4.8.4]{BGL14}) weaker than the
strong log-concavity required in the present work. However, we decided to restrict ourselves to the strong log-concavity condition
since it is easy to check and is commonly used in machine learning and optimisation.

The first and probably the most influential work providing probabilistic analysis of asymptotic properties of the LMC algorithm
is \citep{RobertsTweedie96}. However, one of the recommendations made by the authors of that paper is to avoid using Langevin algorithm
as it is defined in (\ref{algoLMC}), or to use it very cautiously, since the ergodicity of the corresponding Markov chain is very
sensitive to the choice of the parameter $h$. Even in the cases where the Langevin diffusion is geometrically ergodic, the inappropriate
choice of $h$ may result in the transience of the Markov chain $\{\bvartheta^{(k,h)}\}$. These findings have very strongly influenced 
the subsequent studies since all the ensuing research focused essentially on the Metropolis adjusted version of the LMC, known as
Metropolis adjusted Langevin algorithm (MALA), and its modifications
\citep{RobertsRosenthal98,StramerTweedie99-1,StramerTweedie99-2,Jarner2000,RobertsStramer02,Pillai2012,Xifara14}.

In contrast to this, we show here that under the strong convexity assumption imposed on $f$ (or, equivalently, on $-\log \pi$) coupled
with the Lipschitz continuity of the gradient of $f$, one can ensure the non-transience of the Markov chain $\bvartheta^{(k,h)}$
by simply choosing $h\le 1/M$. In fact, the non-explosion of this chain follows from the following proposition the proof of which is
very strongly inspired by the one of Theorem~\ref{th:1}.

\begin{proposition}\label{prop:1}
Let the function $f$ be continuously differentiable on $\RR^p$ and satisfy (\ref{convex1}) with $f^* = \inf_{\bx\in\RR^p} f(\bx)$. Then, for every $h\le 1/M$, we have
\begin{equation}\label{in:prop1}
\bfE \big[f(\bvartheta^{(k,h)})-f^*\big] \le (1-mh)^k\bfE \big[f(\bvartheta^{(0)})-f^*\big] + \frac{Mp}{m}.
\end{equation}
\end{proposition}

Note that under the condition $h\le 1/M$, the quantity $1-mh$ is always nonnegative. Indeed, it follows (see Lemma~\ref{lem:5} in Section~\ref{sec:proofs2}) from the Taylor expansion and
the Lipschitz continuity of the gradient $\nabla f$ that $f(\btheta)-f(\bar\btheta)-\nabla f(\bar\btheta)^\top(\btheta-\bar\btheta)
\le \frac{M}{2}\|\btheta-\bar\btheta\|_2^2$ for every $\btheta,\bar\btheta\in\RR^p$, which---in view of (\ref{convex1})---entails that $m\le M$ and, therefore, $1-m h\ge 1-M h\ge 0$.
On the other hand, in view of the strong convexity of $f$, inequality (\ref{in:prop1}) implies that
\begin{equation}\label{eq:normtheta}
\bfE \big[\|\bvartheta^{(k,h)}-\btheta^*\|_2^2\big] \le \frac{M}{m}\bfE \big[\|\bvartheta^{(0)}-\btheta^*\|_2^2\big] + \frac{2Mp}{m^2},
\end{equation}
where $\btheta^*$ stands for the point of (global) minimum of $f$. As a consequence, the sequence $\bvartheta^{(k,h)}$ produced by the LMC algorithm
is bounded in $L^2$ provided that $h\le 1/M$.

A crucial step in analyzing the long-time behaviour of the LMC algorithm is the assessment of the distance between the distribution of
the random variable $\bL_{Kh}$ and that of $\bvartheta^{(K,h)}$. It is intuitively clear that for a fixed $K$ this distance should
tend to zero when $h$ tends to zero. However, in order to get informative bounds we need to quantify the rate of this convergence.
To this end, we follow the ideas presented in \citep{colt_DalalyanT09,DalalyanTsybakov12a} which consist in performing the following 
two steps. First, a continuous-time Markov process $\{\bD_t:t\ge 0\}$ is introduced such that the distribution of the random vectors $\big(\bvartheta^{(0)},\bvartheta^{(1,h)},\ldots,\bvartheta^{(K,h)}\big)$ and $\big(\bD_{0},\bD_{h},\ldots,\bD_{Kh}\big)$
coincide. Second, the distance between the distributions of the variables $\bD_{Kh}$  and $\bL_{Kh}$ is bounded from above by the 
distance between the distributions of the continuous-time processes $\{\bD_t:t\in[0,Kh]\}$ and $\{\bL_t:t\in[0,Kh]\}$.

To be more precise, we introduce a diffusion-type continuous-time process $\bD$ obeying the following stochastic differential equation:
\begin{equation}\label{D-SDE}
d\bD_t = \boldb_t(\bD)\,dt + \sqrt{2} \; d\bW_t,\qquad t\ge 0,\qquad \bD_0 = \bvartheta^{(0)},
\end{equation}
with the (nonanticipative) drift $\boldb_t(\bD) =  - \sum_{k=0}^{\infty} \nabla f(\bD_{kh})\ind_{[kh,(k+1)h[}(t)$.
By integrating the last equation on the interval $[kh,(k+1)h]$, we check that the increments of this process satisfy
$\bD_{(k+1)h}-\bD_{kh} = -h\nabla f(\bD_{kh}) + \sqrt{2h} \bzeta^{(k+1)}$, where $\bzeta^{(k+1)} = (\bW_{(k+1)h}-\bW_{kh})/\sqrt{h}$.
Since the Brownian motion is a Gaussian process with independent increments, we conclude that   $\{\bzeta^{(k)}:k=1,\ldots,K\}$ is
a sequence of iid standard Gaussian random vectors. This readily implies the equality of the distributions of the random vectors $\big(\bvartheta^{(0)},\bvartheta^{(1,h)},\ldots,\bvartheta^{(K,h)}\big)$ and $\big(\bD_{0},\bD_{h},\ldots,\bD_{Kh}\big)$.

Note that the specific form of the drift $\boldb$ used in the LMC algorithm has the advantage of meeting the following two conditions.
First, $\boldb_t(\bL)$ is close to $-\nabla f(\bL_t)$, the drift of the Langevin diffusion. Second, it is possible to sample from the
distribution $\bPDh(\bx,\:\cdot\:)$, where $h$ is the step of discretisation used in the LMC algorithm. Any nonanticipative drift
function satisfying these two conditions may be used for defining a version of the LMC algorithm. Such an example, the LMC algorithm
with Ozaki discretisation, is considered in Section~\ref{Ozaki}.

To close this section, we state an inequality that will be repeatedly used in this work and the proof of which---based on the
Girsanov formula---can be found, for instance, in \citep{DalalyanTsybakov12a}.
If for some $B >0$ the nonanticipative drift function $\boldb:C(\RR_+,\RR^p)\times \RR_+\to \RR^p$  satisfies the inequality
$\|\boldb(\bD,t)\|_2\le B\big(1+\|\bD\|_\infty\big)$ for every $t\in[0,Kh]$ and every $\bD\in C(\RR_+,\RR^p)$, then
the Kullback-Leibler divergence between $\PP_L^{\bx,Kh}$ and $\PP_D^{\bx,Kh}$, the distributions of the processes $\big\{\bL:t\in[0,Kh]\big\}$ and
$\big\{\bD:t\in[0,Kh]\big\}$ with the initial value $\bL_0=\bD_0=\bx$, is given by
\begin{equation}\label{eq:KL}
\text{KL}\big(\PP_L^{\bx,Kh}\|\PP_D^{\bx,Kh}\big) = \frac14\int_0^{Kh} \bfE\big[\|\nabla f(\bD_t)+\boldb_t(\bD)\|_2^2\big]\,dt.
\end{equation}
It is worth emphasising that the last inequality remains valid when the initial values of the processes $\bD$ and $\bL$ are random but have
the same distribution.

Note that the idea of discretising the diffusion process in order to approximately sample from its invariant density is not new.
It can be traced back at least to \citep{Lamberton2}, see also the thesis \citep{LemairePHD} for an overview. The results therein are
stated for more general discretisation with variable step-sizes but are of asymptotic nature. This point of view has been adopted
and extended to the nonasymptotic case in the recent work \citep{Moulines2015}.

\section{Nonasymptotic bounds on the error of the LMC algorithm}

We are now in a position to establish a nonasymptotic bound with explicit constants on the distance between the
target distribution $\bP_\pi$ and the one produced by the LMC algorithm. As explained earlier, the bound is obtained by controlling two types
of errors: the error of approximating $\bP_\pi$ by the distribution of the Langevin diffusion $\bL_{Kh}$ (\ref{L-SDE}) and the error of
approximating the Langevin diffusion by its discretised version $\bD$ given by  (\ref{D-SDE}). The first error is a decreasing function of $T=Kh$:
in order to make this error small it is necessary to choose a large $T$. A rather precise quantitative assessment of this error is given by
Lemma~\ref{lem:1} in the previous section. The second error vanishes when the step-size $h$ goes to zero, provided that $T=Kh$ is fixed.
Thus, it
is in our interest to choose a small $h$. However, our goal is not only to minimise the error, but also to reduce, as much as possible, the
computational cost of the algorithm. For a fixed $T$, if we choose a small value of $h$ then a large number of steps $K$ is necessary for
getting close to the target distribution. Therefore, the computational complexity is a decreasing function of $h$. In order to find a value of
$h$ leading to a reasonable trade-off between the computational complexity and the approximation error, we need to complement Lemma~\ref{lem:1} with a
precise bound on the second approximation error. This is done in the following lemma.

\begin{lemma}\label{lem:3}
Let $f:\RR^p\to\RR$ be a function satisfying the second inequality in (\ref{convex1}) and $\btheta^*\in\RR^p$ 
be a stationary point (\textit{i.e.}, $\nabla f(\btheta^*)=0$). For any $T>0$, let $\PP_L^{\bx,T}$ and $\PP_D^{\bx,T}$
be respectively the distributions of the Langevin diffusion (\ref{L-SDE}) and its approximation (\ref{D-SDE}) 
on the space of all continuous paths on $[0,T]$ with values in $\RR^p$, with a fixed initial value $\bx$. 
Then, if $h\le 1/(\alpha M)$ with $\alpha\ge 1$, it holds that
\begin{equation}\label{eq:KL1}
\text{\rm KL}\big(\PP_L^{\bx,Kh}\|\PP_D^{\bx,Kh}\big) \le  \frac{M^3h^2\alpha}{12(2\alpha-1)} (\|\bx-\btheta^*\|_2^2+2Khp)+\frac{pKM^2h^2}{4}.
\end{equation}
\end{lemma}

Let us set $T= Kh$. Since it simplifies the mathematical formulae and is possible to achieve in practice in view of Theorem~\ref{th:1},
we assume in the sequel that the initial value of the LMC algorithm is drawn at random from the Gaussian distribution with mean $\btheta^*$, 
a stationary point of $f$, and covariance matrix $M^{-1}\bfI_p$. Then, in view of (\ref{eq:KL1}) and the convexity of the Kullback-Leibler 
divergence,  we get (for $\nu = \mathcal N_p(\btheta^*,M^{-1}\bfI_p)$)
\begin{align}\label{eq:KL2}
\text{\rm KL}\big(\nu\PP_L^{T}\|\nu\PP_D^{T}\big) 
        &\le  \frac{pM^2 h^2\alpha}{12(2\alpha-1)}+\frac{pM^3Th^2\alpha}{6(2\alpha-1)}+\frac{pM^2Th}{4}\nonumber\\
        &=\frac{pM^2Th}{4}\bigg(\frac{\alpha}{3K(2\alpha-1)}+\frac{2M h\alpha}{3(2\alpha-1)}+1\bigg)\le \frac{pM^2Th\alpha}{2(2\alpha-1)},
\end{align}
for every $K\ge \alpha$ and $h\le 1/(\alpha M)$. We can now state the main result of this section, the proof of which 
is postponed to Section~\ref{sec:proofs2}. 

\begin{theorem}\label{th:2}
Let $f:\RR^p\to\RR$ be a function satisfying (\ref{convex1}) and $\btheta^*\in\RR^p$ be its global minimum point.
Assume that for some $\alpha\ge 1$, we have $ h\le 1/(\alpha M)$ and $K \ge \alpha$. Then, for any time horizon $T=Kh$,
the total variation distance between the target distribution $\bP_\pi$ and the approximation $\nu\bP^{K}_{\bvartheta}$ 
furnished by the LMC algorithm with the initial distribution $\nu = \mathcal N_p(\btheta^*, M^{-1}\bfI_p)$ satisfies
\begin{align}\label{mainbound1}
\big\|\nu\bP_{\bvartheta}^{K}-\bP_\pi\big\|_{\rm TV}
    &\le \frac12 \exp\bigg\{\frac{p}{4}\log \bigg(\frac{M}{m}\bigg)-\frac{Tm}{2}\bigg\}
        +\bigg\{\frac{pM^2Th\alpha}{4(2\alpha-1)}\bigg\}^{1/2}.
\end{align}
\end{theorem}

\begin{remark}
The second term in the right-hand side of \eqref{mainbound1} tends to infinity when the time horizon $T$ 
goes to infinity while the step-size $h$ remains fixed. Since the total variation is always bounded by one, 
the obtained bound is not sharp for large values of $T$. The main reason for this is the fact that we upper
bound the total variation distance by the Kullback-Leibler divergence. Improving this argument in order to
get a tighter upper bound is a challenging open problem. 
\end{remark}

We provide here a simple consequence of the last theorem that furnishes
easy-to-apply rules for choosing the time horizon $T$ and the step-size $h$.

\begin{corollary}\label{cor:1}
Let $p\ge 2$, $f$ satisfy (\ref{convex1}) and $\epsilon\in(0,1/2)$ be a target precision level. Let the time horizon $T$ and the step-size
$h$ be defined by
\begin{equation}\label{eq:h}
T = \frac{4\log\big(1/\eps\big) + p\log\big(M/m\big)}{2m},\qquad  h = \frac{\epsilon^2(2\alpha-1)}{M^2Tp\alpha},
\end{equation}
where $\alpha = (1+MpT\epsilon^{-2})/2$. Then the output of the $K$-step LMC algorithm, with $K=\lceil T/h\rceil$, satisfies
$\big\|\nu\bP_{\bvartheta}^{K}-\bP_\pi\big\|_{\rm TV}\le \epsilon$.
\end{corollary}
\begin{proof}
The choice of $T$ and $h$ implies that the two summands in the right-hand side of (\ref{mainbound1}) are bounded by $\epsilon/2$.
Furthermore, one easily checks that  $\alpha = (1+MpT\epsilon^{-2})/2$ is larger than one and satisfies $h\le 1/(\alpha M)$.
In addition, $K \ge T/h\ge \alpha MT \ge 2\alpha (M/m)\log(1/\epsilon)\ge \alpha \log 4$, which ensures the applicability
of Theorem~\ref{th:2}.
\end{proof}

Let us first remark that the claim of Corollary~\ref{cor:1} can be simplified by taking $\alpha=1$. However, for this value of $\alpha$
the factor $(2\alpha-1)/\alpha$ equals one, whereas for the slightly more complicated choice recommended by Corollary~\ref{cor:1}, this factor
is close to two. In practice, increasing $h$ by a factor $2$ results in halving the running time, which represents a non-negligible gain.

Besides providing concrete and easily applicable guidance for choosing the step of discretisation and the stopping rule
for the LMC algorithm to achieve a prescribed error rate, the last corollary tells us that in order to get an error smaller than
$\epsilon$, it is enough to perform $K = O(T^2p/\eps^2) = O\big(\eps^{-2}(p^3+p\log^2(1/\eps))\big)$ evaluations of the gradient of $f$.
To the best of our knowledge, this is the first result that establishes polynomial in $p$ guarantees for sampling from a log-concave
density using the LMC algorithm. We discuss the relation of this and subsequent results to earlier work in Section~\ref{sec:concl}.

\section{Possible extensions}

In this section,  we state some extensions of the previous results that do not require any major change in the proofs, but
might lead to improved computational complexity or be valid under relaxed assumptions in some particular cases.

\subsection{Improved bounds for a ``warm start''}\label{sec:4.1} The choice of the distribution $\nu$ of the initial value $\btheta^{(0)}$
has a significant impact on the convergence of the LMC algorithm. If $\nu$ is close to $\pi$, smaller number of
iterations might be enough for making the TV-error smaller than $\eps$. The goal of this section is to present
quantitative bounds characterising the influence of $\nu$ on the convergence and, as a consequence, on the computational complexity
of the LMC algorithm.

The first observation that can be readily deduced from (\ref{eq:KL1}) is that for any $h\le 1/(2M)$,
\begin{equation}\label{eq:KL1c}
\text{\rm KL}\big(\nu\PP_L^{T}\|\nu\PP_D^{T}\big) \le  \frac{M^3h^2\bfE_{\bvartheta\sim\nu}[\|\bvartheta-\btheta^*\|_2^2]}{18} +
\frac{pM^2Th}{3}.
\end{equation}
Combining this bound with (\ref{eq:6}), Lemma~\ref{lem:1} and (\ref{eq:8}) we get
\begin{align*}
\big\|\nu\bP_{\bvartheta}^{K}-\bP_\pi\big\|_{\rm TV}
    &\le \frac12 \exp\bigg\{\frac{\log \chi^2(\nu\|\pi)-Tm}{2}\bigg\}
        +\bigg\{\frac{M^3h^2\bfE_{\nu}[\|\bvartheta-\btheta^*\|_2^2]+6pM^2Th}{18}\bigg\}^{1/2}.
\end{align*}
Elaborating on this inequality, we get the following result.

\begin{proposition}
Let $\nu$ be a probability density on $\RR^p$ such that the second-order moment $\mu_2 = \frac{M}{p}\bfE_{\bvartheta\sim\nu}[\|\bvartheta-\btheta^*\|_2^2]$ and
the divergence $\chi^2(\nu\|\pi)$ are finite. Then, the LMC algorithm having $\nu$ as initial distribution and using the time horizon $T$ and step-size
$h$ defined by
\begin{equation}\label{eq:hb}
T = \frac{2\log\big(1/\eps\big) + \log\chi^2(\nu\|\pi)}{m},\qquad  h = \frac{9\epsilon^2}{TM^2p(6+\mu_2)},
\end{equation}
satisfies, for $K = [T/h]\ge 2$, the inequality $\big\|\nu\bP_{\bvartheta}^{K}-\bP_\pi\big\|_{\rm TV}\le \epsilon$.
\end{proposition}

The proof of this proposition is immediate and, therefore, is left to the reader. What we infer from this result is that
the choice of the initial distribution $\nu$ has a strong impact on the convergence of the LMC algorithm. For instance,
if for some specific $\pi$ we are able to sample from a density $\nu$ satisfying, for some $\varrho>0$, the relation
$\chi^2(\nu\|\pi)=O(p^\varrho)$ as $p\to\infty$, then the time horizon $T$ for approximating the target density $\pi$
within $\eps$ is $O(\log (p\vee\eps^{-1}))$ and the step-size satisfies $h^{-1}=O(\eps^{-2}p\log(p\vee\eps^{-1}))$. Thus,
in such a situation, one needs to perform $[T/h] = O(\eps^{-2}p\log^2(p\vee\eps^{-1}))$ evaluations of the gradient
of $f$ to get a sampling density within a distance of $\epsilon$ of the target, which is substantially smaller than
$O\big(\eps^{-2}(p^3+p\log^2(1/\eps))\big)$ obtained in the previous section in the general case.

\subsection{Preconditioning}\label{sec:precond} As it is frequently done in optimisation, one may introduce a
preconditioner in the LMC algorithm in
order to accelerate its convergence. To some extent, it amounts to choosing a definite positive $p\times p$ matrix $\bfA$, called
preconditioner, and applying the LMC algorithm to the function $g(\by) = f(\bfA\by)$. Let $\{\bfeta^{(k,h)}:k\in\NN\}$ be the
sequence obtained by the LMC algorithm applied to the function $g$, that is the density of $\bfeta^{(k,h)}$ is close to
$\pi_g(\by)\propto e^{-g(\by)}$ when $k$ is large and $h$ is small. Then, the sequence $\bvartheta^{(k,h)}= \bfA \bfeta^{(k,h)}$
is approximately sampled from the density $\pi_f(\bx)\propto e^{-f(\bx)}$. This follows from the fact that if $\bfeta\sim \pi_g$
then $\bfA\bfeta\sim \pi_f$. Furthermore, it holds that
$$
\|\bfP_{\bvartheta}^{k}-\bfP_{\pi_f}\|_{\rm TV} = \|\bfP_\bfeta^{k}-\bfP_{\pi_g}\|_{\rm TV},
$$
\textit{i.e.}, the approximation error of the LMC algorithm with a preconditioner $\bfA$ is characterised by Corollary~\ref{cor:1}.
This means that if the function $g$ satisfies condition (\ref{convex1}) with constants $(m_{\bfA},M_{\bfA})$, then the number of steps
$K$ after which the preconditioned LMC algorithm has an error bounded by $\eps$ is given by
$K = (M_{\bfA}/m_{\bfA})^2 p\eps^{-2}\big(2\log(1/\eps)+(p/2)\log(M_{\bfA}/m_{\bfA})\big)^2$. Hence, the preconditioner $\bfA$
yielding the best guaranteed computational complexity for the LMC algorithm is the matrix $\bfA$ minimising the ratio $M_{\bfA}/m_{\bfA}$.

The impact of preconditioning can be measured, for instance, in the case of multidimensional logistic regression considered in
Section~\ref{sec:experiments} below. In this case, the ratio $M_{\bfA}/m_{\bfA}$ is up to some constant factor equal to the condition 
number of the matrix $\bfA\Sigma_\bfX\bfA$, where $\Sigma_\bfX$ is the Gram matrix of the covariates. 

\subsection{Nonstrongly log-concave densities}\label{sec:nonstrong}

Theoretical guarantees developed in previous sections assume that the logarithm of the target density is strongly concave,
cf.\  assumption (\ref{convex1}). However, they can also be used for approximate sampling from a density which is log-concave but
not necessarily strongly log-concave; we call these densities nonstrongly log-concave. The idea is then to approximate the target 
density by a strongly log-concave one and to apply the LMC algorithm to the latter instead of the former one.

More precisely, assume that we wish to approximately sample from a multivariate target density $\pi(\bx)\propto \exp\{-f(\bx)\}$, 
where the function $f:\RR^p\to\RR$ is twice differentiable with Lipschitz continuous gradient (\textit{i.e.}, $f$ satisfies the
second inequality in (\ref{convex1})). Assume, in addition, that for every $R\in[0,+\infty]$ there exists $m_R\ge 0$ such that
$\nabla^2 f(\bx)\succeq m_{R}\bfI_p$ for every $\bx\in B=B_R(\bx_0)=\{\bx\in\RR^p:\|\bx-\bx_0\|_2\le R\}$. Here, $\bX_0$ is an 
arbitrarily fixed point in $\RR^p$. Note that if $m_\infty>0$, then this assumption implies the first inequality in (\ref{convex1}) 
with $m=m_\infty$. The purpose of this subsection is to deal with the case where $m_\infty$ equals 0 or is very small. Let $\gamma>0$ 
be a tuning parameter; we introduce the approximate log-density
\begin{equation}\label{f-gamma}
\bar f(\bx)  = f(\bx) + \frac{\gamma}{2}(\|\bx-\bx_0\|_2-R)^2\ind_{B^c}(\bx).
\end{equation}
This function satisfies both inequalities in (\ref{convex1}) with $\bar m =m_{2R}\wedge(m_{\infty}+0.5\gamma)$ and $\bar M = M+\gamma$.
Let us denote by $\bar\pi$ the density defined by $\bar\pi(\bx)\propto e^{-\bar f(\bx)}$ and by $\bfP_{\bar\pi}$
the corresponding probability distribution on $\RR^p$. Heuristically, it is natural to expect that under some mild assumptions
the distribution $\bfP_{\bar\pi}$ is close to the target $\bfP_\pi$ when $R$ is large and $\gamma$ is small. This claim is made rigorous
thanks to the following result, which is stated in a broad generality in order to be applicable to approximations $\bar f$
that are not necessarily of the form (\ref{f-gamma}).

\begin{lemma}\label{lem:lambda}
%Let $\bar f$ be a smooth convex function satisfying (\ref{convex1}) with constants $m_\gamma$ and $M_\gamma$.
Let $f$ and $\bar f$ be two functions such that $f(\bx)\le \bar f(\bx)$ for all $\bx\in\RR^p$ and both $e^{-f}$ and $e^{-\bar f}$
are integrable. Then the Kullback-Leibler divergence between the distribution $\bfP_{\bar\pi}$
defined by the density $\bar\pi(\bx) \propto e^{-\bar f(\bx)}$ and the target distribution $\bfP_\pi$ can be bounded as follows:
\begin{equation}
\text{\rm KL}\big(\bfP_\pi\|\bfP_{\bar\pi}\big) \le   \frac12\int_{\RR^p}\big(\bar f(\bx)-f(\bx)\big)^2\,\pi(\bx)\,d\bx.
\end{equation}
As a consequence, $\big\|\bfP_{\bar\pi}-\bfP_\pi\big\|_{\rm TV} \le \frac12\|\bar f-f\|_{L^2(\pi)}$.
\end{lemma}
\begin{proof}%[Proof of Lemma~\ref{lem:lambda}]
Using the formula for the Kullback-Leibler divergence, we get
\begin{align}\label{ineq:11}
\text{\rm KL}\big(\bfP_{\pi}\|\bfP_{\bar\pi}\big)
	& = \int_{\RR^p} (\bar f(\bx)-f(\bx))\,\pi(\bx)\,dx+\log\int_{\RR^p} e^{f(\bx)-\bar f(\bx)}\pi(\bx)\,d\bx.%\nonumber\\
	%& = \frac{\gamma}2\int_{\RR^p} \|\bx\|_2^2\,\pi(\bx)\,dx+\log\int_{\RR^p} e^{-\frac{\gamma}{2}\|\bx\|_2^2}\pi(\bx)\,d\bx.
\end{align}
Applying successively the inequalities $\log u \le u-1$ and $e^{-u} \le 1-u+\frac12 u^2$ for every $u\ge 0$, we upper bound the
second term in the right-hand side of (\ref{ineq:11}) as follows:
$$
\log\int_{\RR^p} e^{f(\bx)-\bar f(\bx)}\pi(\bx)\,d\bx\le \int_{\RR^p} e^{f-\bar f}\pi-1
\le -\int_{\RR^p} (\bar f-f)\pi +\frac12\int_{\RR^p} (\bar f-f)^2\pi.
$$
Combining this inequality with (\ref{ineq:11}), we get the first claim. The last claim
of the lemma follows from the Pinsker inequality.
\end{proof}

For $\bar f$ given by (\ref{f-gamma}), we get $\big\|\bfP_{\bar\pi}-\bfP_\pi\big\|_{\rm TV} \le
\frac\gamma4 \big(\int_{B^c}(\|\bx-\bx_0\|_2-R)^4\,\pi(\bx)\,d\bx\big)^{1/2}$. Choosing the parameter $\gamma$ sufficiently
small and the parameter $R$ sufficiently large to ensure that $\|\bfP_{\bar\pi}-\bfP_\pi\|_{\rm TV} \le\eps/2$ and assuming
that $\pi$ has bounded fourth-order moment, we derive from this inequality and Corollary~\ref{cor:1} the following convergence
result for the approximate LMC algorithm.

\begin{corollary}\label{cor:2}
Let $f$ be a twice differentiable function satisfying $m_{R}\bfI_p\preceq \nabla^2 f(\bx) \preceq M\bfI_p$ for every $\bx\in B_{R}(\bx_0)$
and for every $R\in[0,+\infty]$. Let $\epsilon\in(0,1/2)$ be a target precision level. Assume that
for some known value $\mu_{R}$ we have $\int_{B_R(\bx_0)^c} (\|\bx-\bx_0\|_2-R)^4\pi(\bx)\,d\bx\le p^2\mu_{R}^2$
and define $\bar m= m_{2R}\wedge(m_\infty+0.5\gamma)$, $\bar M = M+\gamma$ for some $\gamma \le 2\eps/(p\mu_{R})$.
Set the time horizon $T$ and the step-size $h$ as follows:
\begin{equation}\label{eq:Th}
T = \frac{4\log\big(2/\eps\big) + p\log\big(\bar M/\bar m)\big)}{2\bar m},\qquad
h = \frac{\epsilon^2}{4\bar M^2Tp}.
\end{equation}
Then the output of the $K$-step LMC algorithm (\ref{algoLMC}) applied to the approximation $\bar f$ provided by (\ref{f-gamma}),
with $K=\lceil T/h\rceil$,
% $= O(p^5\eps^{-6}\log^2(p\vee\eps^{-1}))$,
satisfies $\big\|\nu\bP_{\bvartheta}^{K}-\bP_\pi\big\|_{\rm TV}\le \epsilon$.
\end{corollary}

Let us comment this result in the case $R=0$ which concerns nonstrongly log-concave densities. Then the previous result implies that
$K= O(p^5\eps^{-4}\log^2(p\vee\eps^{-1}))$. Clearly, the dependence of $K$ both on the dimension $p$ and on the acceptable error level
$\eps$ gets substantially deteriorated as compared to the strongly log-concave case. Some improvements are possible in specific cases.
First, we can improve the dependence of $K$ on $p$ if we are able to simulate from a distribution $\nu$ that is not too far from
$\bar\pi$ in the sense of $\chi^2$ divergence. More precisely, repeating the arguments of Section~\ref{sec:4.1} we get the following
result: if the initial distribution of the LMC algorithm
satisfies $\chi^2(\nu\|\bar\pi) = O\big((p/\gamma)^\varrho\big)$ for some $\varrho>0$ then one needs at most
$K = O\big(p^3\eps^{-4}\log^2(p/\eps)\big)$ steps of the LMC algorithm for getting an error bounded by $\eps$.
Second, in some cases the dependence of $K$ on $p$ can be further improved
by using a preconditioner and/or by replacing the penalty $\|\bx\|_2^2$ in (\ref{f-gamma}) by $\|\bfM\bx\|_2^2$, where $\bfM$
is a properly chosen $p\times p$ matrix.
%Moreover, we can certainly get a better power of $\eps$ if we replace Lemma~\ref{lem:lambda}
%by a tighter result. Since this improvement leads to more involved formulae, we opted for not formalising it here.

This being said, our intuition is that Corollary~\ref{cor:2} is more helpful in the case of convex functions $f$ that are strongly
convex in a neighbourhood of their minimum point $\btheta^*$. In such a situation, our recommendation is to set $\bx_0=\btheta^*$ and
to choose $R$ by maximising the quantity $\bar m= m_{2R}\wedge(m_\infty+\epsilon/(p\mu_R))$. We showcase this approach in
Section~\ref{sec:experiments} on the example of logistic regression.

Note that the convergence of the MCMC methods for sampling from log-concave densities was also studied in~\citep{brooks1998},
where a strategy for defining the stopping rule is proposed. However, as the computational complexity of that strategy
increases exponentially fast in the dimension $p$, its scope of applicability is limited.

\section{Ozaki discretisation and guarantees for smooth Hessian matrices}\label{Ozaki}

For convex log-densities $f$ which are not only continuously differentiable but also have a smooth Hessian matrix
$\nabla^2 f$, it is possible to take advantage of the Ozaki discretisation \citep{Ozaki92} of the Langevin diffusion which
is more accurate than the Euler discretisation analysed in the foregoing sections. It consists in considering
the diffusion process $\bD^O$ defined by (\ref{D-SDE}) with the drift function
\begin{equation}\label{Ozaki1}
\boldb_t(\bD^O)  = - \sum_{k=0}^{K-1} \big\{\nabla f(\bD_{kh}^O)+\nabla^2 f(\bD_{kh}^O)(\bD_t^O-\bD_{kh}^O)\big\}\ind_{[kh,(k+1)h[}(t),
\end{equation}
where, as previously,  $h$ is the step-size and $K$ is the number of iterations to attain the desired time horizon $T= Kh$. This
expression leads to a diffusion process having linear drift function on each interval $[kh,(k+1)h[$. Such a diffusion admits a closed-form
formula. The resulting MCMC algorithm \citep{StramerTweedie99-2}, hereafter referred to as LMCO algorithm (for Langevin Monte
Carlo with Ozaki discretisation), is defined by  an initial value $\bar\bvartheta^{(0)}$ and the following update rule. For every $k\ge 0$,
we set $\bfH_{k} = \nabla^2 f(\bar\bvartheta^{(k,h)})$, which is an invertible $p\times p$ matrix since $f$ is strongly convex,
and define
\begin{align}
&\bfM_k = \big(\bfI_p-e^{-h\bfH_k}\big)\bfH_k^{-1},\qquad
\bfSigma_k  = \big(\bfI_p-e^{-2h\bfH_k}\big)\bfH_k^{-1}, \\
&\bar\bvartheta^{(k+1,h)} = \bar\bvartheta^{(k,h)}-\bfM_k\nabla f\big(\bar\bvartheta^{(k,h)}\big) +
\bfSigma_k^{1/2}\bxi^{(k+1)},\label{update}
\end{align}
where $\{\bxi^{(k)}:k\in\NN\}$ is a sequence of independent random vectors distributed according to the $\mathcal N_p(0,\bfI_p)$ distribution.
In what follows, for any matrix $\bfM$, $\|\bfM\|$ stands for the spectral norm, that is $\|\bfM\| = \sup_{\|\bv\|_2=1} \|\bfM\bv\|_2$.
\begin{theorem}\label{th:3}
Assume that $p\ge 2$, the function $f:\RR^p\to\RR$ satisfies (\ref{convex1}) and, in addition, the Hessian matrix of $f$ is Lipschitz continuous
with some constant $L_f$: $\|\nabla^2 f(\bx)-\nabla^2 f(\bx')\|\le L_f\|\bx-\bx'\|_2$, for all $\bx,\bx'\in\RR^p$.
Let $\btheta^*$ be the global minimum point of $f$ and $\nu$ be the Gaussian distribution $\mathcal N_p(\btheta^*,M^{-1}\bfI_p)$.
Then, for any step-size $h\le 1/(8M)$ and any time horizon $T=Kh\ge 4/(3M)$, the total variation distance between the target distribution $\bP_\pi$
and the approximation furnished by the LMCO algorithm $\nu\bP_{\bar\bvartheta}^{K}$ with $\bar\bvartheta^{(0)}$ drawn at random from $\nu$
satisfies
\begin{align*}%\label{mainbound1b}
\big\|\nu\bP_{\bar\bvartheta}^{K}-\bP_\pi\big\|_{\rm TV}
    &\le \frac12 \exp\Big\{\frac{p}{4}\log \Big(\frac{M}{m}\Big)-\frac{Tm}{2}\Big\}
        +\Big\{L_f^2 Th^2p^2(0.267M^2hT+0.375)\Big\}^{1/2}\!\!.
\end{align*}
\end{theorem}

The proof of this theorem is deferred to Section~\ref{sec:proofs2}. Let us state now a direct consequence of the last theorem,
which provides sufficient conditions on the number of steps for the LMCO algorithm to achieve a prescribed precision level $\epsilon$.
The proof of the corollary is trivial and, therefore, is omitted.

\begin{corollary}\label{cor:4}
Let $f$ satisfy (\ref{convex1}) with a Hessian that is Lipschitz-continuous with constant $L_f$.
For every $\eps\in(0,1/2)$, if the time horizon $T$ and the step-size $h$ are chosen so that
$$
T\ge \frac{4\log(1/\epsilon)+p\log(M/m)}{2m},\qquad h^{-1}\ge (6L_fMTp\eps^{-1})^{2/3}\bigvee (1.25\sqrt{T} L_fp\eps^{-1})\bigvee (8M),
$$
then the distribution of the outcome of the LMCO algorithm with $K= [T/h]$ steps fulfils
$\big\|\nu\bP_{\bar\bvartheta}^{K}-\bP_\pi\big\|_{\rm TV}\le \epsilon$.
\end{corollary}

This corollary provides simple recommendation for the choice of the parameters $h$ and $T$ in the
LMCO algorithm. It also ensures that for the recommended choice of the parameters, it is sufficient to
perform $K = O\big((p+\log(1/\eps))^{3/2}p\eps^{-1}\big)$ number of steps of the LMCO algorithm in order to
reach the desired precision level $\eps$. This number is much smaller than that provided earlier by Corollary~\ref{cor:1},
which was of order $O\big((p+\log(1/\eps))^2p\eps^{-2}\big)$. However, one should pay attention to the fact that
each iteration of the LMCO requires computing the exponential of the Hessian of $f$ at the current state and,
therefore, the computational complexity of each iteration is usually much larger for the LMCO as compared to the LMC
($O(p^3)$ versus $O(p)$). This implies that the LMCO would most likely be preferable to the LMC only in situations where $p$
is not too large, and the required precision level $\eps$ is very small. For instance, the arguments
of this paragraph advocate for using the LMCO instead of the LMC when $\eps = o(p^{-3/2})$.

This being said, it is worth noting that for some functions $f$ the cost of performing a singular values
decomposition on the Hessian of $f$, which is the typical way of computing the matrix exponential, might be
much smaller than the aforementioned worst-case complexity $O(p^3)$. This is, in particular, the case for
the first example considered in the next section. One can also approximate the matrix exponentials by
matrix polynomials. For second-order polynomials, this amounts to replacing the updates (\ref{update}) by
\begin{align}
&\bar\bvartheta^{(k+1,h)} = \bar\bvartheta^{(k,h)}-h\Big(\bfI_p-\frac12h\bfH_k\Big)\nabla f\big(\bar\bvartheta^{(k,h)}\big) +
\sqrt{2h}\Big(\bfI_p-\frac12h\bfH_k\Big)\bxi^{(k+1)}\label{update1}.
\end{align}
Establishing guarantees for such a modified LMCO is out of scope of the present work. We will limit ourselves to
an empirical assessment of the quality of this approximation on the example of logistic regression considered in
Section~\ref{sec:experiments}.

To close this section, let us remark that in the case a warm start is available, the number of iterations for the
LMCO algorithm to reach the precision $\epsilon$ may be reduced to $O^*(p\epsilon^{-1})$. Indeed, if the $\chi^2$
divergence between the initial distribution and the target is bounded by a quantity independent of $p$, or increasing
not faster than a polynomial in $p$, then the time horizon can be chosen as $O^*(1)$ and the choice of $h$ provided by
Corollary~\ref{cor:4} leads to a number of iterations $K$ satisfying $K = O^*(p\epsilon^{-1})$.

\begin{table}
\caption{\label{tab:1}Number of iterations and running times in Example 1}
%\end{table}
%\begin{table}
\centering%
\fbox{%
\begin{tabular}{l|rrrrrrrr}
\toprule
 & $p = 4 $ & $p = 8$ & $p = 12$ & $p = 16$ & $p = 20$ & $p=30$ & $p=40$ & $p=60$\\
 \midrule
 & \multicolumn{8}{c}{Approximate number of iterates, $K$ (to be multiplied by $10^3$)}\\
\midrule
LMC  & $28$ & $87$ & $184$  & $329$  & $532$ & $1350$ & $2728$ & $7741$\\
LMCO & $1$  & $3$  & $5.4$  &   $9$  & $13.6$& $30$   & $54.9$ & $133$ \\
 \midrule
 & \multicolumn{8}{c}{Running times (in seconds) for $N=10^3$ samples}\\
\midrule
LMC  & $3.44$  & $16.6$ & $54.1$   & $123$  & $238$ & $876$  & $2488$ & $9789$\\
LMCO & $0.18$  & $0.70$ & $ 1.78$  & $3.5$  & $6.4$ & $20.4$  & $53.9$ & $189.1$\\
\bottomrule
\end{tabular}}
\end{table}

\section{Numerical experiments}\label{sec:experiments}

To illustrate the results established in the previous sections, we carried out some experiments on
synthetic data. The experiments were conducted on a HP Elitebook PC with the following configuration:
Intel (R) Core (TM) i7-3687U with 2.6 GHz CPU and 16 GB of RAM. The code, written in Matlab, does not
use parallelisation. We considered two examples; both satisfy all the assumptions required in previous sections.
This implies that Corollaries~\ref{cor:1} and \ref{cor:4} apply and guarantee that the choices of
$h$ and $T$ suggested by these corollaries allow us to generate random vectors having a distribution
which is within a prescribed distance $\epsilon$, in total variation, of the target distribution.

\subsection*{Example 1: Gaussian mixture}
The goal of this first experiment is merely to show on a simple example the validity of our theoretical findings.
That is, we check below that the LMC and the LMCO algorithms with the values of time horizon $T$ and step-size
$h$ recommended by Corollaries~\ref{cor:1} and \ref{cor:4}  produce samples distributed approximately as the target
distribution within a reasonable running time. To this end, we consider the simple task of sampling from the density
$\pi$ defined by
\begin{equation}\label{example:1}
\pi(\bx) = \frac1{2(2\pi)^{p/2}}\bigg(e^{-\|\bx-\ba\|_2^2/2}+e^{-\|\bx+\ba\|_2^2/2}\bigg),\qquad \bx\in\RR^p,
\end{equation}
where $\ba\in\RR^p$ is a given vector. This density $\pi$ represents the mixture with equal weights of two
Gaussian densities $\mathcal N(\ba, \bfI_p)$ and $\mathcal N(-\ba, \bfI_p)$. The function $f$, its gradient and
its Hessian are given by
\begin{align*}
f(\bx) &= \frac12\|\bx-\ba\|_2^2 - \log\big(1+e^{-2\bx^\top\ba}\big),\\
\nabla f(\bx) &= \bx-\ba +2\ba \big(1+e^{2\bx^\top\ba}\big)^{-1},\\
\nabla^2f(\bx) &= \bfI_p - 4\ba\,\ba^\top e^{2\bx^\top\ba}\big(1+e^{2\bx^\top\ba}\big)^{-2}.
\end{align*}
Using the fact that $0\le 4e^{2\bx^\top\ba}\big(1+e^{2\bx^\top\ba}\big)^{-2}\le 1$, we infer that
for $\|\ba\|_2<1$, the function $f$ is strongly convex and satisfies (\ref{convex1}) with
$m = 1-\|\ba\|_2^2$ and $M = 1$. Furthermore, the Hessian matrix is Lipschitz continuous with the constant
$L_f = \frac12\|\ba\|_2^3$. Hence, both algorithms explored in the previous sections, LMC and LMCO, can be
used for sampling from the density $\pi$ defined by (\ref{example:1}). Note also that one can sample directly
from $\pi$ by drawing independently at random a Bernoulli$(1/2)$ random variable $Y$ and a standard Gaussian
vector $\bZ\sim \mathcal N(0,\bfI_p)$ and by computing $\bX = Y\cdot (\bZ-\ba)+(1-Y)\cdot(\bZ+\ba)$. The density
of the random vector $\bX$ defined in such a way coincides with $\pi$. One can check that the unique minimum of
$f$ is achieved at $\btheta^* = c^*\cdot\ba$, where $c^*$ is the unique solution of the equation
$c = 1-2(1+e^{2c\|\ba\|_2^2})^{-1}$. Choosing $\ba$ so that $\|\ba\|_2^2=1/2$, we get $\btheta^* = 0$.

\begin{figure}
\includegraphics[width = \textwidth]{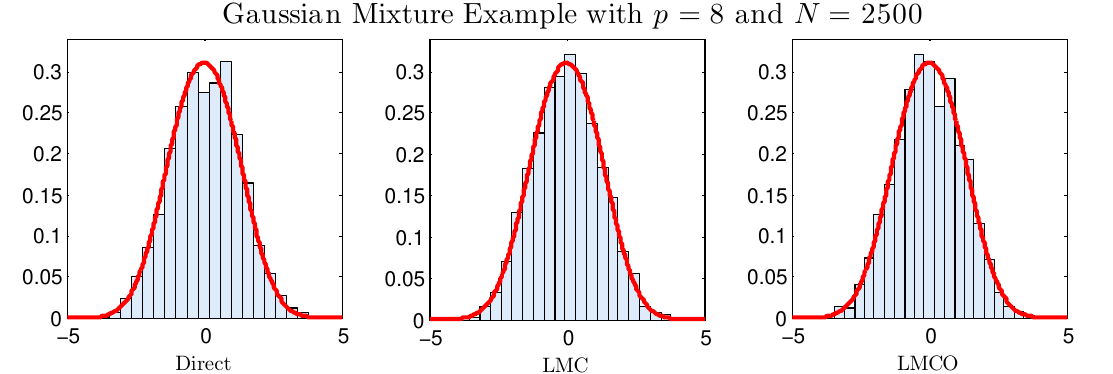}
\vspace{-10pt}
\caption{Histograms of the 1D projections of the samples computed using the Direct (left plot), LMC (middle plot) and
LMCO (right plot) algorithms in the example of a Gaussian mixture (\ref{example:1}). The dimension is $p=8$, the target
precision is $\epsilon=0.1$ and $N=2500$ independent samples were drawn according to each of three methods. The result
shows that both the LMC and the LMCO are very accurate, nearly as accurate as the direct method. \label{fig:1}}
\end{figure}

In the experiment depicted in Figure~\ref{fig:1} (see also Table~\ref{tab:1}), we chose $\epsilon=0.1$ and, for
dimensions $p\in \{4,8,12,16,20,30,40,60\}$, generated vectors using, respectively, the direct method, the
LMC algorithm and the LMCO algorithm. Let $\bvartheta^{{\rm direct},i}$, $\bvartheta^{{\rm LMC},i}$ and
$\bvartheta^{{\rm LMCO},i}$, $i=1,\ldots,N$, be the vectors obtained after  $N$ repetitions of this experiment.
In Figure~\ref{fig:1}, we plotted the histograms of the one-dimensional projections $\bv^\top\bvartheta^{{\rm direct},i}$,
$\bv^\top\bvartheta^{{\rm LMC},i}$ and $\bv^\top\bvartheta^{{\rm LMCO},i}$ of the sampled vectors onto the direction
$\bv = \ba/\|\ba\|_2$ in $\RR^p$ determined by the vector $\ba$. In order to provide a qualitative measure of
accuracy of the obtained samples, we added to each histogram the curve of the true density. The latter can be
computed analytically and is equal to a mixture with equal weights of two one-dimensional Gaussian densities.
The result shows that both the LMC and the LMCO are very accurate, nearly as accurate as the direct method.

To illustrate the dependence on the dimension $p$ of the computational complexity of the proposed sampling
strategies, we report in Table~\ref{tab:1} the number of iterations and the overall running times for generating $N=10^3$
independent samples by the LMC and the LMCO for the target specified by (\ref{example:1}), when the dimension $p$
varies in $\{4,8,12,16,20,30,40,60\}$. One may observe that the computational time is much smaller for the LMCO
than for the LMC algorithm, which is mainly explained by the fact that the singular vectors of the Hessian
of the function $f$, in the example under consideration, do not depend on the value $\bx$ at which the
Hessian is computed.

This example confirms our theoretical findings in that it shows that (a) the samples drawn from the LMC and the LMCO
algorithms with the parameters $T$ and $h$ suggested by theoretical considerations have distributions that are very close
to the target distribution and that (b) the running-times for these algorithms remain reasonable even for moderately large
values of dimension $p$.

\subsection*{Example 2: Binary logistic regression}

Let us consider the problem of logistic regression, in which an iid sample $\{(\bX_i,Y_i)\}_{i=1,\ldots,n}$ is observed, with
features $\bX_i\in\RR^p$ and binary labels $Y_i\in\{0,1\}$. The goal is to estimate the conditional distribution of $Y_1$
given $\bX_1$, which amounts to estimating the regression function $r(\bx) = \bfE[Y_1|\bX_1=\bx] = \bfP(Y_1=1|\bX_1=\bx)$.
In the model of logistic regression, the regression function $r(\bx)$ is approximated by a logistic function of the form
$r(\btheta,\bx) = e^{\btheta^\top\bx}/(1+e^{\btheta^\top\bx})$. The Bayesian approach for estimating the parameter $\btheta$
relies on introducing a prior probability density on $\btheta$, $\pi_0(\cdot)$, and by computing the posterior density $\pi(\cdot)$.
Choosing a Gaussian prior $\pi_0$ with zero mean and covariance matrix proportional to the inverse of the Gram matrix
$\bfSigma_\bfX=\frac1n\sum_{i=1}^n \bX_i\bX_i^\top$,  the posterior density takes the form
\begin{equation}\label{pi:logit}
\pi(\btheta)\propto \exp\Big\{-\bY^\top \bfX\btheta -\sum_{i=1}^n \log(1+e^{-\btheta^\top\bX_i}) -
\frac{\lambda}{2}\|\bfSigma_\bfX^{1/2}\btheta\|_2^2\Big\},
\end{equation}
where $\bY = (Y_1,\ldots,Y_n)^\top \in \{0,1\}^n$ and $\bfX$ is the $n\times p$ matrix having the feature $\bX_i$ as $i^{\rm th}$ row.
The first two terms in the exponential correspond to the log-likelihood of the logistic model, whereas the last term comes from
the log-density of the prior and can be seen as a penalty term. The parameter $\lambda>0$ is usually specified by the practitioner.
Many authors have studied this model from a Bayesian perspective, see for instance \citep{holmes2006,roy2012}, and it seems that
there is no compelling alternative to the MCMC algorithms for computing the Bayesian estimators in this model. Furthermore, even for the
MCMC approach, although geometric ergodicity under some strong assumptions is established, there is no theoretically justified rule for
assessing the convergence and, especially, ensuring that the convergence is achieved in polynomial time. Such guarantees are provided
by our results, when either the LMC or the LMCO is used.

If we define the function $f$ by
\begin{equation}
f(\btheta)= \bY^\top \bfX\btheta + \sum_{i=1}^n \log(1+e^{-\btheta^\top\bX_i}) + \frac{\lambda}{2}\|\bfSigma_\bfX^{1/2}\btheta\|_2^2,
\end{equation}
we get the setting described in the Introduction. It is useful here to apply the preconditioning technique of Section~\ref{sec:precond}
with the preconditioner $\bfA = \bfSigma_\bfX^{-1/2}$. Thus, the LMC and the LMCO can be used with the function $f$ replaced by
$g(\btheta) = f(\bfA\btheta)$. One checks that $g$ and $f$ are infinitely differentiable and
\begin{align*}
\nabla f(\btheta)   &=  \bfX^\top\bY - \sum_{i=1}^n \frac{\bX_i}{1+e^{\btheta^\top\bX_i}} + \lambda\bfSigma_\bfX\btheta,
\quad\nabla^2 f(\btheta) =  \sum_{i=1}^n \frac{e^{\btheta^\top\bX_i}}{(1+e^{\btheta^\top\bX_i})^{2}}\bX_i\bX_i^\top + \lambda\bfSigma_\bfX.
\end{align*}
For the function $g$, since $\nabla^2g(\btheta) = \bfA\nabla^2f(\bfA\btheta)\bfA$, we can infer from these relations that
(\ref{convex1}) holds with $m_\bfA = \lambda$ and $M_\bfA = \lambda+0.25 n $. Note here that if we do not use any preconditioner, the
constants $m$ and $M$ would be given by $m=\lambda\cdot\nu_{\min}(\bfSigma_\bfX)$  and $M = (\lambda+0.25n)\cdot\nu_{\max}(\bfSigma_\bfX)$,
where $\nu_{\min}(\bfSigma)$ and $\nu_{\max}(\bfSigma)$ are respectively the smallest and the largest eigenvalues of $\bfSigma$. This
implies that the ratio $\nu_{\max}(\bfSigma_\bfX)/\nu_{\min}(\bfSigma_\bfX)$ quantifies the gain of efficiency obtained by preconditioning.
This ratio might be large especially when $p$ is large and the covariates are strongly correlated.

Furthermore, $\nabla^2 g$ is Lipschitz with a constant $L_g$ provided by the following formula (the proof
of which is postponed to Section~\ref{sec:proofs2}):
\begin{equation}\label{lf:logit}
L_g = 0.1\Big\|\sum_{i=1}^n \|\bfA\bX_i\|_2 \bfA\bX_i\bX_i^\top\bfA \Big\|\le 0.1n\max_{i=1,\ldots,n} \|\bfSigma_\bfX^{-1/2}\bX_i\|_2.
\end{equation}
%Thus, both the LMC and the LMCO algorithms can be used for sampling from the target density (\ref{pi:logit}).

In our second  experiment, for a set of values of $p$ and $n$, we randomly drew $n$ iid samples $(\bX_i, Y_i)$ according to
the following data generating device. The features $\bX_i$ were drawn from a Rademacher distribution (\textit{i.e.}, each coordinate takes
the values $\pm 1$ with probability $1/2$), and then renormalised to have an Euclidean norm equal to one. Each label $Y_i$, given $\bX_i = \bx$, 
was drawn from a Bernoulli distribution with parameter $r({\btheta^{\rm true}},\bx)$. The true vector $\btheta^{\rm true}$ was set to
${\bf 1}_p = (1,1,\ldots,1)^\top$. For each value of $p$ and $n$, we generated $100$ samples $(\bfX,\bY)$. For each sample, we computed the
MLE using the gradient descent as described in Theorem~\ref{th:1} with a precision level $\epsilon = 10^{-6}$.
%We also computed the Bayesian
%posterior mean and median based on 200 random samples generated by the LMC and the LMCO algorithms with $\epsilon = 0.1$. The step-size $h$ and the time
%horizon $T$ employed in this experiment are those recommended by our theoretical results.
Following the recommendation of \citep{Hanson}, the parameter $\lambda$ was set to $3p/\pi^2$. We carried out two sub-experiments
with well specified distinct purposes: to empirically assess the gain obtained by applying the trick of strong-convexification described in
Subsection~\ref{sec:nonstrong} and to evaluate the loss of accuracy caused by applying to the LMCO algorithm the second-order approximation
(\ref{update1}).
\begin{table}
\caption{\label{tab:2} Example 2 (Binary logistic regression): Number of iterates using the LMC algorithm ($K$) and its modified version as described
in Subsection~\ref{sec:nonstrong} ($K'$).}
\centering
\fbox{%
\begin{tabular}{ll|rrrrr}
\toprule
 & & \multicolumn{5}{c}{$\epsilon = 0.1$}\\
 & & $n = 500 $ & $n = 1000 $ & $n = 2000 $ & $n = 4000 $ & $n = 8000$ \\
\midrule
$p=2$  & $K$     &    $0.065\times 10^7$ &   $0.137\times 10^7$ &   $0.286\times 10^7$ &   $0.596\times 10^7$ &   $1.236\times 10^7$\\
       & $K'$    &    $2.823\times 10^2$ &   $0.688\times 10^2$ &   $0.230\times 10^2$ &   $0.089\times 10^2$ &   $0.039\times 10^2$\\
\midrule
$p=5$  & $K$     &    $0.358\times 10^6$ &   $0.751\times 10^6$ &   $1.568\times 10^6$ &   $3.257\times 10^6$ &   $6.742\times 10^6$\\
       & $K'$    &    $4.207\times 10^4$ &   $0.222\times 10^4$ &   $0.029\times 10^4$ &   $0.007\times 10^4$ &   $0.003\times 10^4$\\
\midrule
$p=20$ & $K$     &    $0.135\times 10^6$ &   $0.279\times 10^6$ &   $0.579\times 10^6$ &   $1.201\times 10^6$ &   $2.481\times 10^6$\\
       & $K'$    &    $0.121\times 10^6$ &   $0.250\times 10^6$ &   $0.519\times 10^6$ &   $1.075\times 10^6$ &   $2.222\times 10^6$\\
\bottomrule
 & & \multicolumn{5}{c}{$\epsilon = 0.01$}\\
 & & $n = 500 $ & $n = 1000 $ & $n = 2000 $ & $n = 4000 $ & $n = 8000$ \\
\midrule
$p=2$  & $K$     &    $0.151\times 10^9$ &   $0.313\times 10^9$ &   $0.645\times 10^9$ &   $1.324\times 10^9$ &   $2.714\times 10^9$\\
       & $K'$    &    $1.529\times 10^5$ &   $0.248\times 10^5$ &   $0.077\times 10^5$ &   $0.028\times 10^5$ &   $0.011\times 10^5$\\
\midrule
$p=5$  & $K$     &    $0.075\times 10^9$ &   $0.155\times 10^9$ &   $0.320\times 10^9$ &   $0.657\times 10^9$ &   $1.345\times 10^9$\\
       & $K'$    &    $3.652\times 10^7$ &   $0.087\times 10^7$ &   $0.011\times 10^7$ &   $0.002\times 10^7$ &   $0.001\times 10^7$\\
\midrule
$p=20$ & $K$     &    $0.254\times 10^8$ &   $0.518\times 10^8$ &   $1.062\times 10^8$ &   $2.177\times 10^8$ &   $4.459\times 10^8$\\
       & $K'$    &    $0.227\times 10^8$ &   $0.463\times 10^8$ &   $0.947\times 10^8$ &   $1.941\times 10^8$ &   $3.975\times 10^8$\\
\bottomrule
\end{tabular}}
\end{table}

In the first sub-experiment, we applied the strategy outlined in Subsection~\ref{sec:nonstrong} for various values of $n,p$ and $\epsilon$.
To this end, we exploited the following formulae
\begin{align*}
m_R &= \lambda + \nu_{\min}(\bfB_R),\quad\text{with}\quad \bfB_R:=\sum_{i=1}^n \frac{e^{|\bX_i^\top\bfA\btheta^*|+R\|\bfA\bX_i\|_2}}
{(1+e^{2|\bX_i^\top\bfA\btheta^*|+2R\|\bfA\bX_i\|_2})^2}\bfA\bX_i\bX_i^\top\bfA,\\
(p\mu_R)^2 & =
\frac{2(M/2)^{p/2}}{(m_RR^2)^{p+4}\Gamma(p/2)}\sum_{j=0}^4C_4^j(-m_RR^2)^j\Gamma(p+4-j;m_RR^2),
\end{align*}
where $\Gamma(p;x) = \int_x^\infty t^{p-1} e^{-t}\,dt$ is the upper incomplete gamma function and $C_4^j$ stands for the binomial coefficient.
The proof of the fact that the quantities $m_R$ and $\mu_R$ defined by these formulae satisfy all the assumptions of Subsection~\ref{sec:nonstrong}
is provided in the supplementary material. In this experiment, we used two values of $\epsilon$ ($0.1$ and $0.01$), three values of dimension $p$
($2$, $5$ and $20$), and five values for the sample size $n$ (500, 1000, 2000, 4000 and 8000). We reported in Table~\ref{tab:2} the number of
iterates using the LMC algorithm ($K$) and the average number of iterates of the modified LMC algorithm as described in Subsection~\ref{sec:nonstrong} ($K'$).
Note that in the case of modified LMC algorithm, the number of iterates depends on the original data $(\bfX,\bY)$. Therefore, the numbers $K'$
reported in Table~\ref{tab:2} are those obtained by averaging over 100 independent trials.

The results of Table~\ref{tab:2} show clearly the advantage of using the strong-convexification trick. For instance, when $\epsilon=0.1$, $p=5$
and $n=1000$, the gain is very impressive since the number of iterations is reduced from nearly $7.5\times10^5$ to $2.2\times10^3$. This represents
a reduction by a factor close to 340. The gain is less significant in the case when the ratio $p/n$ is larger. Our explanation of this phenomenon
is that for a small ratio $p/n$, the posterior density has a very strong peak at its mode. Therefore, even for a relatively large radius $R$ the
condition number $M/m_R$ is not too large.  Thus, small $p/n$ is the typical situation in which the strong-convexification trick is
likely to lead to considerable savings in running-time.

In the second sub-experiment, we aimed at verifying the validity of the second-order approximation of the LMCO algorithm, hereafter referred to as
LMCO', obtained by applying the update rule (\ref{update1}). To this end, for $\epsilon=0.1$, $p\in\{2,5,10\}$ and for $n\in\{200,300,400,500\}$,
we generated $N_{\text{MC}}=100$ Monte-Carlo samples using the LMC algorithm and the LMCO' algorithm. To check the closeness
of the distributions of these two $p$-dimensional samples, we compared several aspects of them. More precisely, we compared their marginal means,
marginal medians and marginal quartiles. Mathematically speaking, for each data-set $\mathcal D_{\text{data}} = (\bfX,\bY)$, we generated
$N_{\text{MC}}$ samples $\mathcal D_{\text{MC}}=\{\btheta^{1},\ldots,\btheta^{N_{\text{MC}}}\}$ and $\bar{\mathcal D}_{\text{MC}}=
\{\bar\btheta^{1},\ldots,\bar\btheta^{N_{\text{MC}}}\}$ using the LMC and the LMCO', respectively. We then computed the normalised
distance between their marginal means: $d_{\text{mean}} = \frac1p\|\mean(\mathcal D_{\text{MC}})-\mean(\bar{\mathcal D}_{\text{MC}})\|_1$.
We also computed the quantities $d_{\rm median}$, $d_{Q_1}$ and $d_{Q_3}$, which are defined analogously by replacing the mean by the 
coordinate-wise median, first quartile and third quartile, respectively. The idea for considering these quantities is that, for large 
$N_{\rm MC}$ and small $\epsilon$, all the aforementioned distances should be close to zero.

We opted for the boxplot representation of 100 values of each of these distances obtained over 100 independent replications of the
data-set $\mathcal D_{\text{data}}$. These boxplots are drawn in Fig.~\ref{fig:2}. They show that the distances are small---at most 
of the order of $10^{-1}$---which may be considered as an argument in favor of the modification proposed in~(\ref{update1}). Indeed, 
with $\epsilon = 0.1$ and $N_{\rm MC}=100$, we could not expect to have an error of smaller order. This is very promising since this 
modified LMCO algorithm has a significantly smaller computational complexity than the original LMCO: each iteration has a worst-case 
accuracy $O(p^2)$ instead of $O(p^3)$, thanks to the fact that matrix exponentials as well as the inversion of the Hessian are replaced 
by the computation of the Hessian and its product with vectors.
\begin{figure}
%\boxed{\vbox{
\includegraphics[width=0.98\textwidth]{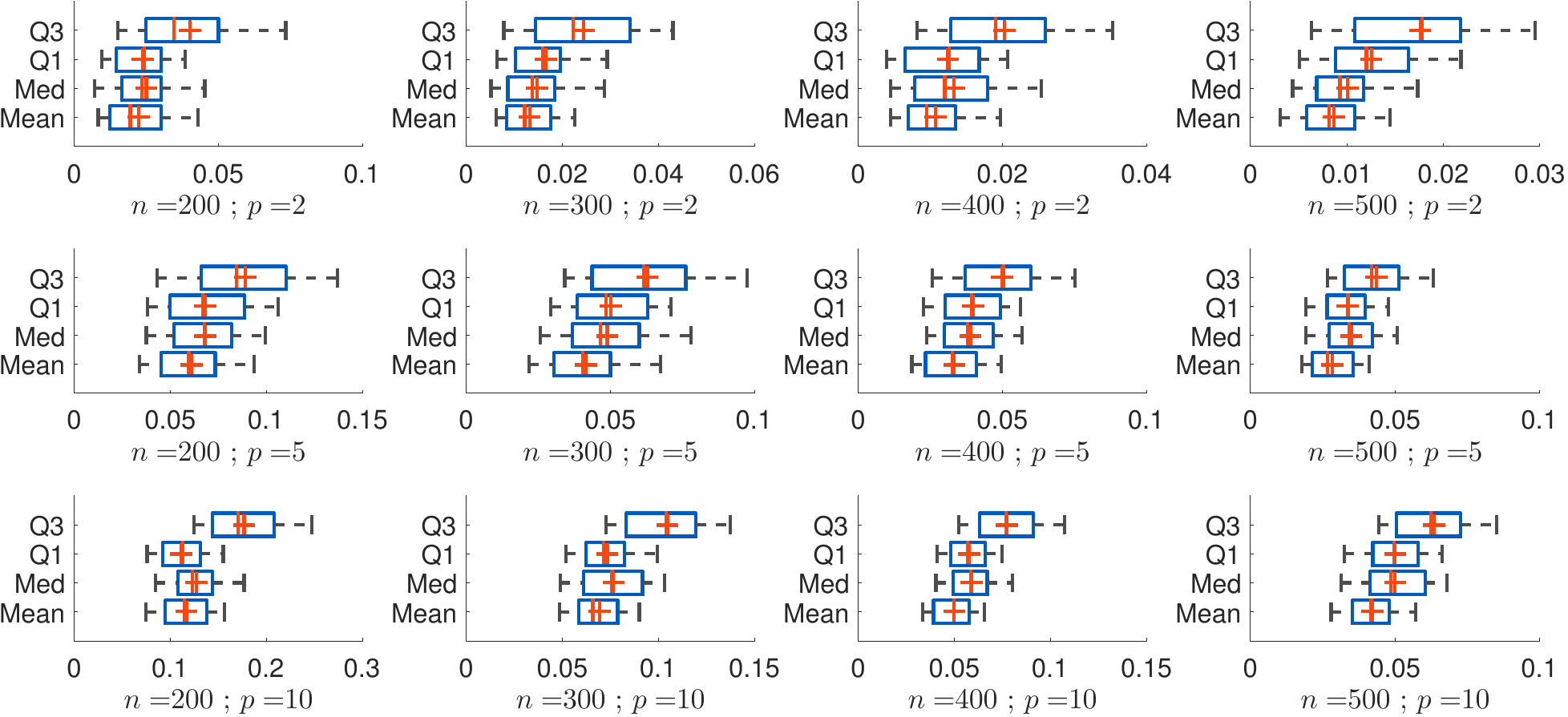}\hspace{-20pt}%}}
\vspace{-10pt}
\caption{Boxplots of the second sub-experiment carried out within the model of logistic regression. }
\label{fig:2}
\end{figure}

%
%\subsection{Application to image denoising}

\section{Summary and conclusion}\label{sec:concl}

We have established easy-to-implement, nonasymptotic theoretical guarantees for approximate sampling from
log-concave and strongly log-concave probability densities. To this end, we have analysed the Langevin Monte Carlo (LMC)
algorithm and its Ozaki discretised version LMCO. These algorithms can be regarded as the natural counterparts---when
the task of optimisation is replaced by the task of sampling---of the gradient descent algorithm, widely studied
in convex optimisation. Despite its broad applicability in the framework of Bayesian
statistics and  beyond, to the best of our knowledge, there were no theoretical result in the literature proving
that the computational complexity of the aforementioned algorithms scales at most polynomially in dimension and in
$\eps^{-1}$, the inverse of the desired precision level. The results
proved in the present work fill this gap by showing that in order to achieve a precision (in total variation) bounded 
from above by $\eps$, the LMC needs no more than $O\big(\eps^{-2}(p^3+p\log(\eps^{-1}))\big)$ evaluations of the gradient
when the target density is strongly log-concave and $O\big(\eps^{-4}p^5\log^2(p\vee\eps^{-1}))\big)$ evaluations of
the gradient when the target density is nonstrongly log-concave.
Further improvement of the rates can be achieved if a ``warm start'' is available. More precisely, if there is an
efficiently samplable distribution $\nu$ such that the chi-squared divergence between $\nu$ and the target scales
polynomially in $p$, then the LMC with an initial value drawn from $\nu$ needs no more than  $O\big(\eps^{-2}p\log^2(p\vee\eps^{-1}))\big)$
evaluations of the gradient when the target density is strongly log-concave and
$O\big(\eps^{-4}p^3\log^2(p\vee\eps^{-1}))\big)$ gradient evaluations when the target density is nonstrongly log-concave.
An important advantage of our results is that all the bounds come with explicit numerical constants of
reasonable magnitude.

The search for tractable theoretical guarantees for MCMC algorithms is an active topic of research not only in probability and
statistics but also in theoretical computer science and in machine learning.  To the best of our knowledge,
first computable bounds on the constants involved in the geometric convergence of Markov chains were derived in
\citep{MeynTweedie94}, see also subsequent work \citep{Rosenthal2002,douc2004} and the survey paper \citep{Roberts2004b}.
However, because of the broad generality of the considered Markov processes\footnote{The authors do not confine their study
to the log-concave densities.}, their results are difficult to implement for getting tight bounds on the constants in the
context of high dimensionality. In particular, we did not succeed in deriving from their results convergence rates for the
LMC algorithm (neither for its Metropolis-Hastings-adjusted version, MALA) that are polynomial in the dimension $p$ and
hold for every strongly log-concave target density. Note also that some nonasymptotic convergence results for the MALA
were obtained by \cite{Hairer2013}, where strongly log-concave four times continuously differentiable functions $f$
were considered. Unfortunately, the constants involved in their bounds are not explicit and cannot be used for our purposes.

The problem of sampling from log-concave distributions is not new. It has been considered in early references \citep{Frieze94}
and \citep{Frieze99}. An important progress in this topic was made in a series of papers by Lov\'azs and Vempala (see,
in particular, \cite{Lovasz2,Lovasz1} for the sharpest results), which are perhaps the closest to our work. They investigated
the problem of sampling from a log-concave density with a compact support and derived nonasymptotic bounds on the number of
steps that are sufficient for approximating the target density; the best bounds are obtained for the hit-and-run algorithm. The analysis they carried out is very different
from the one presented in the present work and the constants in their results are prohibitively large (for instance, $10^{31}$
in \cite[Corollary 1.2]{Lovasz2}), which makes the established guarantees of little interest for practice.
On the positive side, one of the most remarkable features of the results proved in \citep{Lovasz2,Lovasz1}
is that the number of steps required to achieve the level $\eps$ scales polylogarithmically in $1/\eps$.
This is of course much better than the dependence on $\eps$ in our bounds. However, the logarithm of $1/\eps$
in their result is raised to power $5$, which for most interesting values of $\eps$ behaves itself as a linear
function of $1/\eps$. On the down side, the dependence on the dimension in the results of \cite{Lovasz2,Lovasz1},
when no warm start is available, scales as $p^4$, which is worse than $p^3$ inferred from our
analysis. A difference worth being stressed between our framework and that of \cite{Lovasz2,Lovasz1} is that the LMC algorithm
we have analysed here is based on the evaluations of the gradient of $f$, whereas the algorithms studied in \citep{Lovasz2,Lovasz1}
need to sample from the restriction of $\pi_f$ on the lines. On a related note, building on the results by Lov\`azs and Vempala,
\cite{Belloni2009} provided polynomial guarantees for sampling from a distribution which converges asymptotically to a Gaussian one.

After the submission of the present paper, the manuscript \citep{Moulines2015} has been posted on arXiv, which refines our results
in various directions. In particular, the authors of that manuscript manage to assess more accurately the impact of the initial
distribution on the final precision of the LMC algorithm and investigate an Euler scheme with nonconstant step-size. Roughly speaking,
they prove that the rate we obtained in the case of a warm start is valid for any starting point which is not too far away from
the mode of the density. On a related note, we focus in the present work only on the total variation distance
between some MCMC algorithms and the target distribution, whereas in many applications one may be only interested in approximating
integrals with respect to the target distribution. Clearly, guarantees on the total variation distance imply guarantees on the
approximations of integrals, at least when the integrands are bounded functions. However, since the problem of approximating integrals is, in some sense,
easier than sampling from a distribution, one could hope to get tighter bounds for the former problem. This and related
questions are thoroughly investigated in~\citep{Moulines2015}.

Although the main contribution of the present work is of theoretical nature, we can also draw some conclusions which might be
of interest for practitioners. First of all, our results show that the heuristic choice of the stopping rule for the MCMC algorithms
is not the only possible option: it is also possible to have theoretically grounded guidelines for choosing the stopping time.
The resulting algorithm will be of polynomial complexity both in dimension and in the precision level. Second, the results reported
in this work show that there is no need to apply Metropolis-Hastings correction to the Langevin algorithm and its various variants
in order to ensure their convergence. Third, when the dimension is not very high and a high level of precision is required (\textit{i.e.},
when $p^{3/2}\epsilon$ is small), the LMCO algorithm is preferable to the LMC algorithm, and the modified LMCO using the update rule
of Eq.\ (\ref{update1}) is even better. Note, however, that this last claim was checked empirically but comes without any theoretical
justification.

Finally, we would like to mention that, in recent years, several studies making the connection between convex optimisation and
MCMC algorithms were carried out. They mainly focused on proposing new algorithms of approximate
sampling \citep{Girolami11,Schreck2013,Pereyra2014} inspired by the ideas coming from convex optimisation. We hope that the
present work will stimulate a more extensive investigation of the relationship between approximate sampling and optimisation,
especially in the aim of establishing user friendly theoretical guarantees for the MCMC algorithms.

\section{Postponed proofs and some technical results}\label{sec:proofs2}

\subsection{Auxiliary results}
%We start by recalling a simple result. %For the sake of completeness, its proof is also provided.

\begin{lemma}[Lemma 1.2.3 in \cite{Nest}]\label{lem:5}
If the function $f$ satisfies the second inequality in (\ref{convex1}), then
$f(\btheta)-f(\bar\btheta)-\nabla f(\bar\btheta)^\top (\btheta-\bar\btheta)
\le \frac{M}2\|\btheta-\bar\btheta\|_2^2$, $\forall \btheta,\bar\btheta\in\RR^p$.
\end{lemma}
%\begin{proof}
%Let us introduce the auxiliary function $\varphi:[0,1]\to \RR$ defined by $\varphi(t) = f(\bar\btheta + t(\btheta-\bar\btheta))$.
%This definition readily entails that $f(\btheta)-f(\bar\btheta)-\nabla f(\bar\btheta)^\top (\btheta-\bar\btheta)=\varphi(1)-\varphi(0)-\varphi'(0)$.
%Since $\varphi$ is continuously differentiable, it holds $\varphi(1)-\varphi(0)-\varphi'(0) = \int_0^1 (\varphi'(t)-\varphi'(0))\,dt\le
%\int_0^1 |\varphi'(t)-\varphi'(0)|\,dt$. In view of the definition of $\varphi$, we have
%\begin{align*}
%|\varphi'(t)-\varphi'(0)| &= \big|\big(\nabla f(\bar\btheta + t(\btheta-\bar\btheta))-\nabla f(\bar\btheta)\big)^\top(\btheta-\bar\btheta)\big|\\
%&\le \big\|\nabla f(\bar\btheta + t(\btheta-\bar\btheta))-\nabla f(\bar\btheta)\big\|_2\cdot\|\btheta-\bar\btheta\big\|_2 \qquad \text{(by Cauchy-Schwarz)}\\
%&\le M t\|\btheta-\bar\btheta\big\|_2^2. \qquad\qquad\qquad\qquad\qquad\qquad\qquad\text{(since $f$ satisfies (\ref{convex1}))}
%\end{align*}
%This leads to the inequalities
%\begin{align*}
%f(\btheta)-f(\bar\btheta)-\nabla f(\bar\btheta)^\top (\btheta-\bar\btheta)
%        &\le \int_0^1 |\varphi'(t)-\varphi'(0)|\,dt\\
%        &\le \int_0^1 M t\|\btheta-\bar\btheta\big\|_2^2\,dt
%        \le \frac{M}{2}\|\btheta-\bar\btheta\big\|_2^2,
%\end{align*}
%which complete the proof.
%\end{proof}

\begin{lemma}\label{lem:2}
Let us denote by $\nu_{h,\bx}$ the conditional density of $\bvartheta^{(1,h)}$ given $\bvartheta^{(0)}=\bx$, where the sequence
$\{\bvartheta^{(k,h)}\}_{k\in\NN}$ is defined by (\ref{algoLMC}) with a function $f$ satisfying (\ref{convex1}). (In other terms,
$\nu_{h,\bx}$ is the density of the Gaussian distribution $\mathcal N(\bx-h\nabla f(\bx),2h\bfI_p)$.) If $h\le 1/(2M)$ then
$$
\bfE_\pi\bigg[\frac{\nu_{h,\bx}(\bvartheta)^2}{\pi(\bvartheta)^2}\bigg]
     \le  \exp\bigg\{\frac{1}{2m}\|\nabla f(\bx)\|_2^2-\frac{p}{2}\log (2 hm)\bigg\}.
$$
\end{lemma}
\begin{proof}
%Let us bound from above the term $\bfE_\pi\big[\frac{\nu_{h,\bx}(\bvartheta)^2}{\pi(\bvartheta)^2}\big]$ arising from the application of Lemma~\ref{lem:1}.
In view of the relations
\begin{align*}
\pi(\btheta)^{-1}
    & = e^{f(\btheta)}\int_{\RR^p} e^{-f(\bar\btheta)}\,d\bar\btheta
		  = e^{f(\btheta)-f(\bx)}\int_{\RR^p} e^{-f(\bar\btheta)+f(\bx)}\,d\bar\btheta\\
    &\le e^{\nabla f(\bx)^\top(\btheta-\bx)+\frac{M}{2}\|\btheta-\bx\|_2^2}\int_{\RR^p}
		     e^{-\nabla f(\bx)^\top(\bar\btheta-\bx)-\frac{m}{2}\|\bar\btheta-\bx\|_2^2}\,d\bar\btheta\\
    &\le \bigg(\frac{2\pi}{m} \bigg)^{p/2}\exp\Big\{{\nabla f(\bx)^\top(\btheta-\bx)+\frac{M}{2}\|\btheta-\bx\|_2^2+\frac{1}{2m}\|\nabla f(\bx)\|_2^2}\Big\}
\end{align*}
we have
\begin{align*}
\bfE_\pi\bigg[\frac{\nu_{h,\bx}(\bvartheta)^2}{\pi(\bvartheta)^2}\bigg]
    &  =  (4\pi h)^{-p} \int_{\RR^p} \exp\Big\{ -\frac{1}{2h}\;\|\btheta-\bx+h\nabla f(\bx)\|_2^2\Big\}\,\pi(\btheta)^{-1}\,d\btheta    \\
    & \le (4\pi h)^{-p}(2\pi /m )^{p/2} e^{\frac{1}{2m}\|\nabla f(\bx)\|_2^2}
		      \int_{\RR^p} \exp\Big\{ -\frac{(1-hM)\|\btheta-\bx\|_2^2}{2h}\Big\}\,d\btheta\\
    & =   (4\pi h)^{-p}(2\pi/m)^{p/2} (2\pi h)^{p/2}(1-hM)^{-p/2}e^{\frac{1}{2m}\|\nabla f(\bx)\|_2^2}.
\end{align*}
After a suitable rearrangement of the terms we get the claim of Lemma~\ref{lem:2}.
\end{proof}

\subsection{Proofs of results concerning the LMC}

Instead of proving Proposition~\ref{prop:1}, we prove below the following stronger result.

\begin{proposition}\label{prop:1b}
Let the function $f$ be continuously differentiable on $\RR^p$ and satisfy (\ref{convex1}) with $f^* = \inf_{\bx\in\RR^p} f(\bx)$. Then, for every $h\le 1/M$, we have
\begin{align}
\bfE \big[f(\bvartheta^{(k,h)})-f^*\big] &\le (1-mh)^k\bfE \big[f(\bvartheta^{(0)})-f^*\big] + \frac{Mp}{m(2-Mh)},\label{in:prop1:a}\\
\bfE \big[\|\bvartheta^{(k,h)}-\btheta^*\|_2^2\big] &\le \frac{Me^{-mhk}}{m}\bfE \big[\|\bvartheta^{(0)}-\btheta^*\|_2^2\big] + \frac{2Mp}{m^2(2-Mh)}.\label{eq:normtheta:b}
\end{align}
\end{proposition}

\begin{proof}%[Proof of Proposition~\ref{prop:1b}]
Throughout this proof, we use the shorthand notation $f^{(k)} =f(\bvartheta^{(k,h)})$ and
$\nabla f^{(k)}= \nabla f(\bvartheta^{(k,h)}) $. In view of the relation (\ref{algoLMC}) and the Taylor expansion, we have
\begin{align*}
f^{(k+1)} &\le  f^{(k)} +(\nabla f^{(k)})^\top (\bvartheta^{(k+1,h)}-\bvartheta^{(k,h)}) +
\frac{M}{2} \|\bvartheta^{(k+1,h)}-\bvartheta^{(k,h)}\|_2^2\\
& = f^{(k)} -h\|\nabla f^{(k)} \|_2^2  + \sqrt{2h}\; (\nabla f^{(k)})^\top\bxi^{(k+1)}+
\frac{M}{2} \|h\nabla f^{(k)} -\sqrt{2h}\;\bxi^{(k+1)}\|_2^2.
\end{align*}
Taking the expectations of both sides, we get
\begin{align}\label{eq:2}
\bfE\big[f^{(k+1)} \big] &\le \bfE\big[f^{(k)} \big] -h\bfE\big[\|\nabla f^{(k)} \|_2^2\big]
+\frac{M}{2} h^2 \bfE\big[\|\nabla f^{(k)} \|_2^2\big]+Mhp\nonumber\\
&=
\bfE\big[f^{(k)} \big] -\frac12h(2-Mh)\bfE\big[\|\nabla f^{(k)} \|_2^2\big]+Mhp.
\end{align}
It is well known (see, for instance,  \citep{BoydBook}) that for the global minimum $f^*$ of $f$ over $\RR^p$, we have
$$
\|\nabla f(\bx)\|_2^2 \ge 2m \big(f(\bx)-f^*\big),\qquad \forall \bx\in\RR^p.
$$
Applying this inequality to $\bx = \bvartheta^{(k,h)}$ and combining it with (\ref{eq:2}), whenever $h< 2/M$ we get
\begin{align}\label{eq:3}
\bfE\big[f^{(k+1)} \big] &\le
\bfE\big[f^{(k)} \big] -mh(2-Mh)\bfE\big[f^{(k)} -f^*\big]+Mhp.
\end{align}
Let us set $\gamma = mh(2-Mh)\in(0,1)$ for any $h\in(0,2/M)$. Subtracting $f^*$ from the both sides of (\ref{eq:3}) we arrive at
\begin{align}\label{eq:4}
\bfE\big[f^{(k+1)} -f^*\big] &\le (1-\gamma)\bfE\big[f^{(k)} -f^*\big]+Mhp.
\end{align}
This implies that
\begin{align}\label{eq:5}
\bfE\big[f^{(k+1)} -f^*\big] &\le (1-\gamma)^{k+1}\bfE\big[f(\bvartheta^{(0)})-f^*\big]+Mhp(1+\ldots+(1-\gamma)^{k})\nonumber\\
&\le(1-\gamma)^{k+1}\bfE\big[f(\bvartheta^{(0)})-f^*\big]+Mhp\gamma^{-1}.
%\nonumber\\
%&=(1-\gamma)^{k+1}\bfE\big[f(\bvartheta^{(0)})-f^*\big]+\frac{Mp}{m(1-Mh)}.
\end{align}
Inequality (\ref{in:prop1:a}) follows by replacing $\gamma$ by $mh(2-Mh)$. To prove (\ref{eq:normtheta:b}), it suffices to combine
(\ref{in:prop1:a}) with the first inequality in (\ref{convex1}), Lemma~\ref{lem:5} and the inequality $(1-mh)^k\le e^{-mhk}$.
\end{proof}

\begin{corollary}\label{cor:3}
Let $h\le 1/\alpha M$ with $\alpha\ge 1$ and $K\ge 1$ be an integer. Under the conditions of Proposition~\ref{prop:1}, it holds
$$
h\sum_{k=0}^{K-1} \bfE[\|\nabla f(\bvartheta^{(k,h)})\|_2^2] \le \frac{M\alpha}{2\alpha-1}\bfE\big[\|\bvartheta^{(0)}-\btheta^*\|_2^2\big]+ \frac{2\alpha M Khp}{2\alpha -1}.
$$
\end{corollary}
\begin{proof}
Using inequality (\ref{eq:2}) and the fact that $2-Mh\ge (2\alpha-1)/\alpha$, we get
$$
\frac{h(2\alpha-1)}{2\alpha}\bfE\big[\|\nabla f^{(k)} \|_2^2\big]  \le
\bfE\big[f^{(k)}-f^{(k+1)}\big] +Mhp,\qquad \forall k\in\NN.
$$
Summing up these inequalities for $k=0,\ldots,K-1$ and using the obvious bound $f^{(K)}\ge f^*$, we get
$$
h\sum_{k=0}^{K-1}\bfE\big[\|\nabla f^{(k)} \|_2^2\big]  \le  \frac{2\alpha}{2\alpha-1}\bfE\big[f^{(0)}-f^*\big] +\frac{2\alpha MKhp}{2\alpha-1}.
$$
To complete the proof, it suffices to remark that in view of Lemma~\ref{lem:5}, it holds
$2\bfE\big[f^{(0)}-f^*\big]\le M\bfE\big[\|\bvartheta^{(0)}-\btheta^*\|_2^2\big]$.
\end{proof}

\begin{proof}[Proof of Lemma~\ref{lem:1}]
The first inequality in (\ref{convex1}) yields $\big(-\nabla f(\btheta)+\nabla f(\bar\btheta)\big)^\top(\btheta-\bar\btheta)\le -\frac{m}{2}\|\btheta-\bar\btheta\|_2^2$
for every $\btheta,\bar\btheta\in\RR^p$. Therefore, according to \cite[Remark 4.14]{ChenWang} and \cite[Corollary 4.8.2]{BGL14}, the process $\bL_t$ is geometrically ergodic
in $L^2(\RR^p,\pi)$  that is:
\begin{equation}\label{spectralgap}
\int_{\RR^p} \big( \bfE\big[\varphi(\bL_t)|\bL_0=\bx\big]-\bfE_\pi\big[\varphi(\bvartheta)\big]\big)^2 \pi(\bx)\,d\bx \le e^{-t m} \bfE_\pi\big[\varphi^2(\bvartheta)\big]
\end{equation}
for every $t>0$ and every $\varphi\in L^2(\RR^p;\pi)$. The claim of the lemma follows from this inequality by simple application of the Cauchy-Schwarz inequality.
Indeed, by definition of the total variation and in view of the fact that $\pi$ is the invariant density of the semigroup $\bPLt$, we have
\begin{align*}
\|\nu\bPLt-\pi\|_{\rm TV}
&=\sup_{A\in\mathscr B(\RR^p)} \bigg|\int_{\RR^p}\bPLt(\bx,A)\nu(\bx)\,d\bx-\pi(A)\bigg|\\
&=\sup_{A\in\mathscr B(\RR^p)} \bigg|\int_{\RR^p}\big(\bPLt(\bx,A)-\pi(A)\big)\nu(\bx)\,d\bx\bigg|\\
&=\sup_{A\in\mathscr B(\RR^p)} \bigg|\int_{\RR^p}\big(\bPLt(\bx,A)-\pi(A)\big)\big(\nu(\bx)-\pi(\bx)\big)\,d\bx\bigg|\\
&\le\sup_{A\in\mathscr B(\RR^p)} \int_{\RR^p}\Big|\bPLt(\bx,A)-\pi(A)\Big|\cdot\Big|\frac{\nu(\bx)}{\pi(\bx)}-1\Big|\,\pi(\bx)\,d\bx.
\end{align*}
Using the Cauchy-Schwarz inequality, we get
\begin{align*}
\|\nu\bPLt-\pi\|_{\rm TV}
    &\le\sup_{A\in\mathscr B(\RR^p)} \bigg(\int_{\RR^p}\big|\bPLt(\bx,A)-\pi(A)\big|^2\,\pi(\bx)\,dx\bigg)^{1/2}
        \chi^2(\nu\|\pi)^{1/2}.
\end{align*}
For every fixed Borel set $A$, if we set $\varphi(\bx) = \ind_A(\bx)-\pi(A)$ and use (\ref{spectralgap}), we obtain that
\begin{align*}
\int_{\RR^p}\big|\bPLt(\bx,A)-\pi(A)\big|^2\,\pi(\bx)\,d\bx
& = \int_{\RR^p}\big(\bfE\big[\varphi(\bL_t)|\bL_0=\bx\big]-\bfE_\pi\big[\varphi(\bvartheta)\big]\big)^2\,\pi(\bx)\,d\bx \\
&\le e^{-tm} \bfE_\pi\big[\varphi(\bvartheta)^2\big]\\
&= e^{-tm} \pi(A)(1-\pi(A))\le \frac14\; e^{-tm}.
\end{align*}
This completes the proof of the lemma.
\end{proof}

\begin{proof}[Proof of Lemma~\ref{lem:3}]
Setting $T=Kh$ and using (\ref{eq:KL}), we get
\begin{align*}
\text{KL}\big(\PP_{\bL}^{\bx,T}\|\PP_{\bD}^{\bx,T}\big)
&= \frac14 \int_0^{T} \bfE\big[\|\nabla f(\bD_t)+\boldb_t(\bD)\|_2^2\big]\,dt\\
&= \frac14 \sum_{k=0}^{K-1}\int_{kh}^{(k+1)h} \bfE\big[\|\nabla f(\bD_t)-\nabla f(\bD_{kh})\|_2^2\big]\,dt.
\end{align*}
Since $\nabla f$ is Lipschitz continuous with Lipschitz constant $M$, we have
\begin{align*}
\text{KL}\big(\PP_{\bL}^{\bx,T}\|\PP_{\bD}^{\bx,T}\big)
&\le \frac{M^2}4 \sum_{k=0}^{K-1}\int_{kh}^{(k+1)h} \bfE\big[\|\bD_t-\bD_{kh}\|_2^2\big]\,dt.
\end{align*}
In view of (\ref{D-SDE}) we obtain
\begin{align}
\text{KL}\big(\PP_{\bL}^{\bx,T}\|\PP_{\bD}^{\bx,T}\big)
&\le \frac{M^2}4 \sum_{k=0}^{K-1}\int_{kh}^{(k+1)h}\Big( \bfE\big[\|\nabla f(\bD_{kh})\|_2^2(t-kh)^2\big]+2p(t-kh)\Big)\,dt\nonumber\\
&= \frac{M^2h^3}{12} \sum_{k=0}^{K-1}\bfE\big[\|\nabla f(\bvartheta^{(k,h)})\|_2^2\big]+\frac{pKM^2h^2}{4}.
\end{align}
Applying Corollary~\ref{cor:3}, the desired inequality follows.
\end{proof}

\begin{proof}[Proof of Theorem~\ref{th:2}]
In view of the triangle inequality, we have
\begin{align}\label{eq:6}
\|\nu\bP_{\bvartheta}^{K}-\bP_\pi\|_{\rm TV}
        & = \|\nu\bP_{\bD}^{Kh}-\bP_\pi\|_{\rm TV}
				 \le \|{\nu}\bP_{\bL}^{T}-\bP_\pi\|_{\rm TV}+ \|\nu\bP_{\bD}^{T}-\nu\bP_{\bL}^{T}\|_{\rm TV} .
\end{align}
The first term in the right-hand side is what we call first type error. The source of this error is
the finiteness of time, since it would be equal to zero if we could choose
$T=Kh=+\infty$. The second term in the right-hand side of (\ref{eq:6}) is the second type error, which
is caused by the practical impossibility to take the step-size $h$ equal to zero. These two errors can be evaluated as follows.

For the first type error, apply Lemma~\ref{lem:1} to get
$\|\nu\bP_{\bL}^{T}-\bP_\pi\|_{\rm TV}\le  \frac12\chi^2(\nu\|\pi)^{1/2} e^{-T m/2}$.
Since $\nu$ is a Gaussian distribution, the expectation in the above formula is not difficult to evaluate.
The corresponding result, provided by Lemma~\ref{lem:2}, yields
\begin{align}\label{eq:7}
\|{\nu}\bP_{\bL}^{T}-\bP_\pi\|_{\rm TV}&\le \frac12 \exp\bigg\{\frac{p}{4}\log \bigg(\frac{M}{m}\bigg)-\frac{Tm}{2}\bigg\}.
\end{align}
To evaluate the second type error, we use the Pinsker inequality:
\begin{align}\label{eq:8}
\|\nu\bP_{\bD}^{T}-\nu\bP_{\bL}^{T}\|_{\rm TV}&\le \|\nu\PP_{\bD}^{T}-\nu\PP_{\bL}^{T}\|_{\rm TV}
\le \Big\{\frac12 \text{KL}\big(\nu\PP_{\bL}^{T}\|\nu\PP_{\bD}^{T}\big)\Big\}^{1/2}.
\end{align}
Combining this inequality with (\ref{eq:KL2}), we get the desired result.
\end{proof}
\subsection{Proofs of results concerning the LMCO}

\begin{proof}[Proof of Theorem~\ref{th:3}]
Using the same arguments as those of the proof of Theorem~\ref{th:2}. This leads to the inequality
\begin{align}\label{pr:th3:1}
\big\|\nu\bP_{\bvartheta}^{K}-\bP_\pi\big\|_{\rm TV}
    &\le \frac12 \exp\bigg\{\frac{p}{4}\log \bigg(\frac{2M}{m}\bigg)-\frac{Tm}{2}\bigg\}
        +\bigg\{\frac12 \text{KL}\big(\nu\PP_{\bL}^{T}\|\nu\PP_{\bD^O}^{T}\big)\bigg\}^{1/2},
\end{align}
where $\PP_{\bD^O}^{T}$ is the probability distribution induced by the diffusion process $\bD^O$
corresponding to the Ozaki discretisation (in fact, it is a piecewise Ornstein-Uhlenbeck process). Relation (\ref{eq:KL})
implies that
\begin{align}\label{pr:th3:2}
\text{KL}\Big(\nu\PP_{\bL}^{T}\Big\|\nu\PP_{\bD^O}^{T}\Big) =
\frac14 \int_0^T \bfE\Big[\big\|\nabla f(\bD_t^O)+b_t(\bD^O)\big\|_2^2\Big]\,dt.
\end{align}
Since on each interval $[kh,(k+1)h[$ the function $t\mapsto b_t$ is linear, for every $t\in [kh,(k+1)h[$, we get
%\begin{align}\label{pr:th3:3}
$\big\|\nabla f(\bD_t^O)+b_t(\bD^O)\big\|_2^2
		= \big\|\nabla f(\bD_t^O)-\nabla f(\bD_{kh}^O)-\nabla^2 f(\bD_{kh}^O)(\bD^O_t-\bD_{kh}^O)\big\|_2^2$.
%\end{align}
Using the mean-value theorem and the Lipschitz continuity of the Hessian of $f$, we derive from the above relation that
\begin{align}\label{pr:th3:4}
\big\|\nabla f(\bD_t^O)+b_t(\bD^O)\big\|_2^2
		&\le \frac14L_f^2 \big\|\bD^O_t-\bD_{kh}^O\big\|_2^4,
\end{align}
for every $t\in [kh,(k+1)h[$. Note now that equation (\ref{update}) provides the conditional distribution of $\bD_{(k+1)h}^O$
given $\bD_{kh}^O$. An analogous formula holds for the conditional distribution of $\bD_{t}^O-\bD_{kh}^O$ given $\bD_{kh}^O$,
which is multivariate Gaussian with mean $\big(\bfI_p-e^{-(t-kh)\bfH_k}\big)\bfH_k^{-1}\nabla f\big(\bD_{kh}^O\big)$
and covariance matrix  $\bfSigma_k =\big(\bfI_p-e^{-2(t-hk)\bfH_k}\big)\bfH_k^{-1}$, where
$\bfH_k = \nabla^2 f(\bD_{kh}^O)$.  Under convexity condition on $f$, we have $\|\big(\bfI_p-e^{-s\bfH_k}\big)\bfH_k^{-1}\|\le s$
for every $s>0$.  Therefore, conditioning with respect to
$\bD_{kh}^O$ and using the inequality $(a+b)^4\le 8(a^4+b^4)$, for every $t\in[kh,(k+1)h[$ we get
\begin{align*}%\label{pr:th3:5}
\frac14\bfE\big[\big\|\bD^O_t-\bD_{kh}^O\big\|_2^4\,\big|\,\bD_{kh}^O\big]
		& \le  \big\|\big(\bfI_p-e^{-(t-kh)\bfH_k}\big)\bfH_k^{-1}\nabla f\big(\bD_{kh}^O\big)\big\|^4_2+
		 \bfE\Big[\big\|\bfSigma_k^{1/2}\bxi^{(k+1)}\big\|^4_2\,\Big|\,\bD_{kh}^O\Big]\\
		& \le  (t-hk)^4\big\|\nabla f\big(\bD_{kh}^O\big)\big\|^4_2+  (p+1)^2
		\big\|(\bfI_p-e^{-2(t-hk)\bfH_k})\bfH_k^{-1}\big\|^2\\
		& \le  (t-hk)^4\big\|\nabla f\big(\bD_{kh}^O\big)\big\|^4_2+ 4 (t-hk)^2 (p+1)^2.
\end{align*}
This inequality, in conjunction with (\ref{pr:th3:2}) and (\ref{pr:th3:4}) yields
\begin{align}\label{pr:th3:5}
\text{KL}\Big(\nu\PP_{\bL}^{T}\Big\|\nu\PP_{\bD^O}^{T}\Big)
		& \le \frac{L_f^2}{16} \sum_{k=0}^{K-1}\int_{kh}^{(k+1)h} \bfE\Big(\bfE\Big[\big\|\bD_t^O-\bD_{kh}^O\big\|_2^4
				\,\big|\bD_{kh}^O\Big]\Big)\,dt\nonumber\\
		& \le \frac{L_f^2 h^5}{20}\sum_{k=0}^{K-1} \bfE\big(\big\|\nabla f\big(\bD_{kh}^O\big)\big\|^4_2\big)+ \frac{1}{3} L^2_f Kh^3 (p+1)^2.
\end{align}
%Using the second inequality in (\ref{convex1}) and the fact that $\nabla f(\btheta^*)=0$ we can upper bound the
%term $\|\nabla f\big(\bD_{kh}^O\big)\big\|_2$ by $M\|\bD_{kh}^O-\btheta^*\|_2 = M\|\bar\bvartheta^{(k,h)}-\btheta^*\|_2$.
%The latter norm can be bounded from above using the next lemma, the proof of which is postponed.
%
%\begin{lemma}\label{lem:4}
%If $p\ge 2$ and $T = Kh\ge 8/(3M)$, then the iterates of the LMCO algorithm satisfy
%\begin{align*}
%\bfE\big[\|\bar\bvartheta^{(K,h)}-\btheta^*\|_2^4\big]\le 7\bigg(\frac{4TMp}{m}\bigg)^2.
%\end{align*}
%\end{lemma}

To bound the last expectation, we use the fact that $\bD_{kh}^O$ equals  $\bar\bvartheta^{(k,h)}$ in distribution, and
the next lemma (the proof of which is provided in the supplementary material).

\begin{lemma}\label{lem:4}
If $p\ge 2$, $T \ge 4/(3M)$ and $h\le 1/(8M)$, then the iterates of the LMCO algorithm satisfy
$\bfE\big[\big(\sum_{k=0}^{K-1}\|\nabla f(\bar\bvartheta^{(k,h)})\|_2^2\big)^2\big] \le \frac{32}{3}\big({TMp}/h\big)^2$.
\end{lemma}

Combining this lemma and (\ref{pr:th3:5}), we upper bound the Kullback-Leibler divergence as follows
$\text{KL}\big(\nu\PP_{\bL}^{T}\big\|\nu\PP_{\bD^O}^{T}\big)
		 \le 0.534  h^3(L_f TMp)^2+ 0.75  T (L_f h p)^2$, which completes the proof.
\end{proof}

{\renewcommand{\addtocontents}[2]{}
\section*{Acknowledgments}
The work of the author was partially supported by the grant Investissements d'Avenir (ANR-11-IDEX-0003/Labex Ecodec/ANR-11-LABX-0047).
}

{\renewcommand{\addtocontents}[2]{}
\bibliography{Literature}}

\end{document}